\documentclass[journal]{IEEEtran}
\usepackage{cite}  
\usepackage{amsmath,amssymb,amsfonts}
\usepackage{mathrsfs}
\usepackage{amsfonts}
\usepackage{makecell}
\usepackage{textcomp}
\usepackage{xcolor}    
\usepackage{amsthm}  
\usepackage{graphicx}  
\usepackage{multirow}  
\usepackage{algorithm}
\usepackage{algorithmicx}
\usepackage[noend]{algpseudocode}
\usepackage{tikz}
\usepackage{verbatim}
\usepackage{graphics}
\usepackage{epsfig}
\usepackage{lineno,hyperref}
\usepackage{array}
\usepackage{booktabs}
\usepackage{bm}
\usepackage{tikz}
\usepackage{url}
\usepackage{bbding} 
\usepackage{diagbox}
\usepackage{threeparttable}
\mathchardef\mhyphen="2D
\newtheorem{theorem}{Theorem}
\newtheorem{definition}{Definition}

\begin{document}

\title{DSMIX: A Dynamic Self-organizing Mix Anonymous System}
\author{
	\IEEEauthorblockN{Renpeng Zou\IEEEauthorrefmark{1}, Xixiang Lv\IEEEauthorrefmark{1},} 
	
	\IEEEauthorblockA{\IEEEauthorrefmark{1}School of Cyber Engineering, Xidian University, Xian 710071, China} 

}
\maketitle

\begin{abstract}
Increasing awareness of privacy-preserving has led to a strong focus on anonymous  systems protecting anonymity. By studying early schemes, we summarize some intractable problems of anonymous systems. Centralization setting is a universal problem since most anonymous system rely on central proxies or presetting nodes to forward and mix messages, which compromises users' privacy in some way. Besides, availability becomes another important factor limiting the development of anonymous system due to the large requirement of additional additional resources (i.e. bandwidth and storage) and high latency. Moreover, existing anonymous systems may suffer from different attacks including abominable Man-in-the-Middle (MitM) attacks, Distributed Denial-of-service (DDoS) attacks and so on. In this context, we first come up with a BlockChain-based Mix-Net (BCMN) protocol and theoretically demonstrate its security and anonymity. Then we construct a concrete dynamic self-organizing BlockChain-based MIX anonymous system (BCMIX). In the system, users and mix nodes utilize the blockchain transactions and their addresses to negotiate keys with each other, which can resist the MitM attacks. In addition, we design an IP sharding algorithm to mitigate Sybil attacks. To evaluate the BCMIX system, we leverage the distribution of mining pools in the real world to test the system's performance and ability to resistant attacks. Compared with other systems, BCMIX provides better resilience to known attacks, while achieving low latency anonymous communication without significant bandwidth or storage resources.
\end{abstract}


\IEEEpeerreviewmaketitle

\section{Introduction}
\IEEEPARstart{K}{eeping} communication private has become increasing important in an era of mass surveillance and carriers-sponsored attacks. Recently, many events about private data leakage in Online Social Networks (OSNs) \cite{sanders2019facebook}, mobile network service (Uber, Didi Chuxing) \cite{hayes2017geolocation} and telephone communications \cite{alexopoulos2017mcmix} have occurred frequently. Thus it is often the case that two parties want to communicate anonymously, which means to exchange messages while hiding the fact that they are in conversation.

In this context, the anonymous communication technology emerges as a critical topic. Aiming to preserve communication privacy within the shared public network environment, anonymous communication mainly focus on how to hide the identities or address information of one side or both sides in communications. Since the seminal work by Chaum \cite{chaum1981untraceable} for anonymous communication, more than seventy anonymous systems have been proposed, based on different anonymous mechanism \cite{lu2019survey}. Generally, anonymous systems can be divided into the following sub-types, mix re-encryption, multicast/broadcast, mix multi-layer encryption and peer-to-peer. 

There are, however, concerns about the lack of efficiency and security in anonymous systems \cite{danezis2008survey},\cite{edman2009anonymity}, as explained below: 
\begin{itemize}
\item {\textbf{Centralization.}} Most of anonymous systems, i.e. Anonymizer \cite{boyan1997anonymizer}, LPWA \cite{gabber1999consistent} and cMix \cite{chaum2017cmix}, rely on central proxies or preset nodes to hide the address information, then still use these proxies to blind and forward messages. The centralized anonymous systems bring privacy leakage risks to users since the service providers can control all proxies and mix node to infer the users' identities. What's worse, the public proxies or preset nodes are easily exposed to attackers. For instance, an attacker can launch distributed denial-of-service attacks to block one or more proxies and thus crash the system.

\item{\textbf{Availability.}} Efficiency and additional resources are the main factors affecting the availability of anonymous systems. High-latency anonymous systems such as OneSwarm \cite{prusty2011forensic} and A3 \cite{sherr2010a3}, though provide high anonymity, are not well suited for practical utilization because of the poor efficiency in terms of intolerable latency \cite{prusty2011forensic}. In order to hide the identities of the recipients, multicast/broadcast and peer-to-peer based anonymous systems consume more bandwidth resources to cover the normal traffics \cite{kong2003anodr}. In spite of effectiveness, users may not be willing to contribute a lot of bandwidth, which hinders the development of such anonymous systems. In addition, some anonymous systems, such as cMix, adopt additional inspection schemes to identify the malicious nodes, which would place an additional burden on users and decrease message utilization. 

\item{\textbf{Security.}} Along the research line about security, anonymous systems based on different mechanism may suffer from different security issues. Some anonymous systems based on MIX technologies rely on fixed cascade nodes to mix and forward messages. Such systems are vulnerable to collusion-tagging attacks which is hard to detect. Re-routing or proxy forwarding based anonymous systems, such as Tor \cite{syverson2004tor} and Tarzan \cite{freedman2002tarzan}, deliver messages through nodes or proxies randomly selected from clusters, hence an eavesdropper can perform traffic analysis attacks and destroy the anonymity. As for multicast/broadcast based anonymous system, an attacker can masquerade as the benign recipients to intercept the messages. With respect to the P2P based
anonymous systems, an attacker can create multiple identities to launch a Sybil attack, which allows the attacker to analyze the forwarding path and impersonate the recipient to receive the messages. Moreover, in anonymous systems applying key exchange schemes, an attacker can employ Man-in-the-Middle (MitM) attacks \cite{conti2016survey} to undermine the security of these systems. 
\end{itemize} 

\subsection{Solutions and Contributions}
Motivated by these identified limitations, we combine blockchain technology with mix network, and design a dynamic self-organizing blockchain-based mix anonymous system. We expect to dispose of the following challenges. 

\textbf{Challenge 1. Designing a decentralized self-organizing anonymous system.} The first challenge we seek to address is the centralization issues. As we mention above, the central proxies or preset nodes might pry into users' private data and reveal the true identities. Therefore, we intend to construct a decentralized anonymous system in which the mix nodes are dynamic and self-organizing.

\textbf{Solution 1.} When it comes to decentralization, the most popular technology is blockchain, an emerging decentralized architecture and distributed computing paradigm underlying Bitcoin \cite{nakamoto2019bitcoin} and other cryptocurrencies. Leveraging the dynamic and decentralized properties of blockchain miners, we devise voting algorithms to elect mix nodes from blockchain miners. This trustless and distributed design not only prevents privacy leakage from service provides, but also mitigates single point failure and DDos attacks.

\textbf{Challenge 2. Designing a user friendly anonymous system.} Another challenge is to construct an anonymous system with high availability. In fact, users are often reluctant to consume more additional resources (i.e. bandwidth and computing resources) and wait for a long time. Thus we are drove to build a high available anonymous system with less additional resources.

\textbf{Solution 2.} Our proposed solution is inspired by cMix. Similar to cMix, we also split the time-consuming, complicated public key operations with the real time phase. The difference is that cMix adopts fixed mix nodes with a stable joint public key while our approach leverages the dynamic blockchain miners to compete for mix nodes, which brings a problem that the elected mix nodes (miners) have to negotiate a joint public key in each round. To avoid interacting with other mix nodes in each round, we firstly propose a revised additive homomorphism mix-net protocol. Then we combine the protocol with Bitcoin account schemes, such that the elected mix nodes can calculate the system public key in a non-interactive manner. It is worth noting that no additional resources are required in BCMIX other than the miners' computing power for solving Bitcoin puzzles.

\textbf{Challenge 3. Designing a secure anonymous system.} The most intractable challenge is building a secure anonymous system. Based on different principles, anonymous systems are assailable to disparate attacks including internal attacks and external attacks. The internal attacks 
, i.e. traffic analysis attacks, DDos attacks and so on, are caused by the design principles of systems while the external attacks, such as MitM attacks and Sybil attacks, arise from the cryptographic protocols or other schemes utilized in anonymous systems. In the case, we intend to construct an anonymous system which can resist the attacks mentioned above.

\textbf{Solution 3.} Through the former schemes we found mix technology can defend against most kind of attacks on anonymous systems except bring the centralization problem and tagging attacks. To mitigate the weakness, our first consideration is combining blockchain and mix technology, as mentioned in Solution 1. Unfortunately, the introduction of blockchain raises the Sybil attacks into the anonymous system. By researching the properties of Sybil attacks we structure PoW voting and IP sharding algorithms to mitigate the impact of Sybil attacks. For MitM attacks, we design a transaction-based key exchange scheme which makes use of the Bitcoin's transaction propagation mechanism to break the single-channel control of attackers. The detailed illustration is provided in Section V and Section VI.

To summarize, the contributions of this paper are as follows.
\begin{itemize}
	\item We propose a blockchain-based mix-net protocol (BCMN) whose security and anonymity are demonstrated theoretically. Especially, we elect miners as mix nodes via special algorithms to avoid the centralized settings. 
	\item We construct a dynamic self-organizing blockchain-based mix anonymous system (upon the proposed BCMN protocol) and discuss how the proposal can satisfy the security requirements. 
	\item We demonstrate the feasibility and effectiveness of the proposed BCMIX by developing the system in an analog network with the miner distribution data in the real world. Compared with existing systems, our system performs well and provides stronger security.
\end{itemize}
\subsection{Related Works}
Generally, anonymous communication systems can be divided into the following sub-types, mix re-encryption, multicast/broadcast, mix multi-layer encryption and peer-to-peer. Mix re-encryption based anonymous systems \cite{gomulkiewicz2004onions},\cite{pereira2017marked} leverage cryptography technologies to dispose messages and hide the users' identities, which can resist traffic analysis attacks. But with the utilizing of public key cryptography, mix re-encryption based anonymous systems is expensive and easy to waste resources. To settle the problem, Chaum et.al \cite{chaum2017cmix} proposed a anonymous system called cMix in 2017. Through a precomputation, the core cMix protocol eliminates all expensive real time public-key operations—at the senders, recipients and mix nodes, thereby decreasing real time cryptographic latency and lowering computational costs for clients. The core real time phase performs only a few fast modular multiplications. The authors claim that cMix can detect the malicious nodes by utilizing Random Partial Checking (RPC) and commitment scheme.

Multicast/broadcast based anonymous communication systems \cite{chaum1988dining}, \cite{kotzanikolaou2017broadcast} achieve anonymity through one-to-many communications among hosts. This method is expensive and inefficient for non-broadcast networks. In the case of large scale networks, an attacker can easily masquerade as the recipient to intercept the message, which further increases the computation and communication overhead required for authentication. 

As for mix multi-layer encryption based anonymous systems \cite{egger2013practical} \cite{syverson2004tor} \cite{piotrowska2017loopix}, one or more proxies are selected from the cluster, and forward the messages from the former nodes and then the messages in a confusing order. The technology can achieve low-latency communications under the premise of ensuring efficiency, but is vulnerable to analyzing attacks such as traffic analysis attacks and sniper attacks\cite{jansen2014sniper}. 

The rapid development of peer-to-peer (P2P) networks drives the research of anonymous communication technology in P2P network environment \cite{chothia2005survey}. In P2P based anonymous systems \cite{ruffing2017p2p}, \cite{han2008mutual} nodes enjoy anonymous services, and provide anonymous services for other nodes in their spare time. Because the P2P network itself has a high degree of self-organization and disorder, and the number of members is large, the P2P network can also provide a high degree of anonymity when the attacker has a huge attack resource. However, because of its openness and anonymity, the attacker can control a large number of zombie nodes to launch witch attacks, and can disguise as normal nodes for traffic analysis, thereby destroying the system's anonymity without being noticed.
\subsection{Roadmap}
The rest of this paper is organized as follows. In Section II, we review some preliminaries and propose an attacks against cMix. In Section III, we illustrate the security model and requirements. Nest, we detail the BCMN protocol and the proposed BCMIX in Section IV. Then we evaluate the performance and demonstrate security respectively in Section V and Section VI. Finally, this paper is concluded in Section VII.
\section{Preliminaries and Proposed Attacks on cMix}
In this section we briefly review the relevant notations and definitions that are used in this paper. Then we describe some attacks against cMix.
\subsection{Elliptic Curve based Cryptographic Primitives}
In this work, we adopt elliptic curves over prime finite field $\mathbb{F}_p$. The elliptic curve $y^2=x^3+ax+b$ over $\mathbb{F}_p$ could be represented as $E_p(a,b)$.
	
\textbf{EC-Elgamal.} The analog of ElGamal crypto system based on ECC, which is known as EC-Elgamal, was first introduced in \cite{koblitz1987elliptic}. It consists of the following algorithms.
\begin{itemize}
	\item $\mathsf{Setup}(1^{\mathcal{K}}).$ The algorithm takes as input the security parameter $\mathcal{K}$, and outputs the elliptic curve $E_p(a,b)$ with base point $G$.
	\item $\mathsf{KeyGen}.$ For the elliptic curve $E_p(a,b)$ with base point $G$, pick $k {\stackrel{R}{\longleftarrow}}{\mathbb{F}_p}$ and compute $K=kG$. Set $PK=K$ and $SK=k$.
	\item $\mathsf{Enc}(m,PK,r).$ For the plaintext point $m$, pick $r {\stackrel{R}{\longleftarrow}}{\mathbb{F}_p}$. Compute $C_1=rG$, $C_2=m+rK$. Set ciphertext points $C=(C_1,C_2)$.
	\item $\mathsf{Dec}(C,SK).$ For the ciphertext $C$, compute $m'=C_2-kC_1$.
\end{itemize}
\begin{definition}
Suppose $p$ is a prime and $E_p(a,b)$ is an elliptic curve. For the two points $G$ and $Q$ on the elliptic curve, they satisfy $Q=kG$. It can be proved that it is easier to calculate $Q$ from $k$ and $G$. However, it is difficult to calculate $k$ from $Q$ and $G$ \cite{luo2019image}.
\end{definition}

The security of ECC is based on Elliptic Curve Discrete Logarithm Problem (ECDLP) which is consider to be computationally infeasible to solve. 

\textbf{ECDH Key Exchange.} Elliptic Curve Diffie-Hellman (ECDH) key exchange is the elliptic cuive analogue of the classical Diffie-Hellman key exchange operating in $\mathbb{Z}_p^*$. We describe two communicating parties, usually called Alice and Bob, establish a shared secret key in secure communication channel as follows. We assume that Alice and Bob use the same set of domain parameters  $D:=(p,a,b,G,n,h)$ for their computations.
\begin{itemize}
	\item Alice generates an ephemeral key pair $(k_A,Q_A)$, i.e. he/she generates a random number $k_A$ in $[1,n-1]$ and then performs a scalar multiplication to get the corresponding public key $Q_A= k_A \cdot G$. Then Alice sends $Q_A$ to Bob.
	\item Bob generates an ephemeral key pair $(k_B,Q_B)$ with $(Q_B=k_B\cdot G)$ in the same way and sends $Q_B$ to Alice.
	\item After Alice receives $Q_B$, he/she performs a scalar multiplication to obtain the shared secret $S=k_A\cdot Q_B$.
	\item After Bob receives $Q_A$ from Alice, he/she obtains the shared secret through computation of $S=k_B\cdot Q_A$
\end{itemize}

The security of the ECDH protocol relies on the intractability of (computational) Elliptic Curve Diffie-Hellman Problem (ECDHP). That is, given an elliptic curve $E_p(a,b)$, a base point $G\in E(\mathbb{F}_P)$, and two points $Q_A=k_A \cdot G$ and $Q_B=k_B \cdot G$, find the point $S=k_A\cdot k_B\cdot G$ without knowledge of $k_A$, $k_B$. It is clear that an algorithm for solving a generic ECDLP instance would allow one to solve the ECDHP as well.
\subsection{Verifiable Random Function.}
A Verifiable Random Function (VRF) \cite{jager2015verifiable} is the public-key version of a keyed cryptographic hash. In this application, a Prover holds the VRF secret key and uses the VRF hashing to construct a hash-based data structure on the input data. Due to the nature of the VRF, only the Prover can answer queries about whether or not some data is stored in the data structure. Anyone who knows the public VRF key can verify that the Prover has answered the queries correctly. A VRF is a triplet of algorithms $\mathsf{VRF}:=(\mathsf{Gen},\mathsf{Eval},\mathsf{Vfy})$ providing the following functionalities.
\begin{itemize}
	\item $\mathsf{Gen}(1^{\mathcal{K}}).$ The key generation algorithm is a probabilistic algorithm that takes as input the security parameter $\mathcal{K}$ and outputs a key pair $(vk,vsk)$. We say that $vsk$ is the secret key and $vk$ is the verification key.
	\item $\mathsf{Eval}(vsk,X).$ The deterministic algorithm on input the secret key $vsk$ and $X\in \{0,1\}^k$ and outputs a function value $Y\in\mathcal{Y}$, where $\mathcal{Y}$ is a finite set, and a proof $\pi$. We write $V_{vsk}(X)$ to denote the function value $Y$ computed by $\mathsf{Eval}$ on input $(vsk,X)$.
	\item $\mathsf{Ver}(vk,X,Y,\pi).$ The verification algorithm takes as input $(vk,X,Y,\pi)$ and outputs a bit $b\in\{0,1\}$ indicating whether or not $\pi$ is a valid proof.
\end{itemize}

\textbf{Blockchain Basics.} We review some basic components of a proof-of-work blockchain \cite{garay2015bitcoin}. We define a transaction $tx:=( \overrightarrow{inputs},\overrightarrow{outputs},sig)$, where $\overrightarrow{inputs}$ and $\overrightarrow{outputs}$ are the inputs and outputs of a UTXO-based model, $sig$ is the signature signed by the transaction sender. A block is a triple of the form $B:=(s,x,txs,ctr)$, $s\in\{0,1\}^{\mathcal{K}}, x\in\{0,1\}^*,ctr\in \mathbb{N}$, where $s$ is the state of the previous block, $x$ is the data and $ctr$ is the proof of work of the block. A block $B$ is valid iff\\
$$\mathsf{validBlock}^D(B):=H(ctr,T(s,x,txs))<D.$$

Here, $H:\{0,1\}^*\rightarrow\{0,1\}^{\mathcal{K}}$ and $T:\{0,1\}^*\rightarrow\{0,1\}^{\mathcal{K}}$ are cryptographic hash functions, and the parameter $D\in \mathbb{N}$ is the difficulty level of the block.

A chain is simply a chain of blocks, that we call $\mathcal{C}$. The rightmost block is called the head of the chain, denoted by $\mathsf{Head}(\mathcal{C})$. Any chain $\mathcal{C}$ with a head $\mathsf{Head}(\mathcal{C}):=(s,x,ctr)$ can be extended to a new longer chain $\mathcal{C}':=\mathcal{C}||B'$ by attaching a block $B':=(s',x',ctr')$ such that $s'=H(ctr,G(s,x))$; the head of the new chain $\mathcal{C}'$ is $\mathsf{Head}(\mathcal{C}):=B'$. We let $\mathcal{C}:=\varepsilon$ to express a chain $\mathcal{C}$ is empty. The function $\mathsf{len}(\mathcal{C})$ denotes the length of a chain $\mathcal{C}$.

For a chain $\mathcal{C}$ of length $n$ and any $q>0$, we denote by $\mathcal{C}^{\ulcorner q}$ the chain resulting from removing the $q$ rightmost blocks of $\mathcal{C}$, and analogously we denote by $\mathcal{C}^{\urcorner q}$ the chain resulting in removing the $q$ leftmost blocks of $\mathcal{C}$; note that if $q\ge n$ then $\mathcal{C}^{\ulcorner q}:=\varepsilon$ and  $\mathcal{C}^{\urcorner q}:=\varepsilon$. If $\mathcal{C}$ is a prefix of $\mathcal{C}'$ we write $\mathcal{C}\prec \mathcal{C}'$. We also leverage $slot$ which is defined in [\cite{}]. A slot is the continuous amount of divided time. Each slot $slot_l$ is indexed for $l\in\{1,2,3,\cdots\}$. We assume that users have a synchronised clock that indicates the current time down to the smallest discrete unit.

\textbf{Blockchain protocol.} With the illustrations of the basic components, we describe the blockchain protocol \cite{garay2015bitcoin} $\Gamma=(\Gamma.\mathsf{KeyGen},\Gamma.\mathsf{Update},\Gamma.\mathsf{Validate},\Gamma.\mathsf{Broadcast})$ as follows.
\begin{itemize}
	\item $\{pk,sk\}\leftarrow\Gamma.\mathsf{KeyGen}$: The algorithm generates the key pair $(pk,sk)$ of the blockchain nodes.
	\item $\{\mathcal{C}',\perp\}\leftarrow \Gamma.\mathsf{Update}$: This algorithm returns a longer and valid chain $\mathcal{C}$ in the network (if it exists), otherwise returns $\perp$.
	\item $\{0,1\}\leftarrow \Gamma.\mathsf{Validate}$: The validity check algorithm takes as inputs a transaction $tx$, a block $B$ or a chain $\mathcal{C}$ and returns 1 iff the transaction, the block or the chain is valid according to a public set of rules.
	\item $\Gamma.\mathsf{Broadcast}$: The algorithm takes as inputs some $tx$s and broadcasts it to all the nodes of the blockchain system.
\end{itemize}

The security of a PoW blockchain protocol $\Gamma$ is characterized by three properties, namely: $Chain$ $Growth$, $Chain$ $Quality$ and $Common$ $Prefix$ \cite{garay2015bitcoin}.

\textbf{Chain growth.} The chain property quantifies the number of blocks that are added to the blockchain during any given number of slots.
\begin{definition}(Chain Growth).
	Consider the chains $\mathcal{C}_1$, $\mathcal{C}_2$ possessed by two honest parties at the onset of two slots $slot_1$, $slot_2$, with $slot_2$ at least $s$ slots ahead of $slot_1$. Then it holds that $\mathsf{len}(\mathcal{C}_2)-\mathsf{len}(\mathcal{C}_1)\ge \tau \cdot s$, for $s\in \mathbb{N}$ and $0<\tau\le1$, where $\tau$ is the speed coefficient.
\end{definition}
 
\textbf{Chain Quality.}The chain quality property informally states that the ratio of adversarial blocks in any segment of a chain held by a honest party is no more than a fraction $\mu$, where $\mu$ is the fraction of resources controlled by the adversary.

\begin{definition}(Chain Quality).
	Consider a portion of length $\ell$-blocks of a chain possessed by an honest party during any given slot intervals, for $\ell\in\mathbb{N}$. Then, the ratio of adversarial blocks in this $\ell$ segment of the chain is at most $\mu$, where $0<\mu\le1$ is the chain quality coefficient.
\end{definition}

\textbf{Common Prefix.} The common prefix property informally says that if we take the chains of two honest nodes at different times slots, the shortest chain is a prefix of the longest chain.  
\begin{definition}(Common Prefix).
	The chains $\mathcal{C}_1$, $\mathcal{C}_2$ possessed by two honest parties at the onset of the slots $slot_1<slot_2$ are such that $\mathcal{C}_1^{\ulcorner k}\preceq \mathcal{C}_2$, where $\mathcal{C}_1^{\ulcorner k}$ denotes the chain obtained by removing the last $k$ blocks from $\mathcal{C}_1$, where $k\in \mathbb{N}$ is the common prefix parameter.
\end{definition}	
\subsection{cMix Anonymous System}
The cMix protocol by Chaum et al.\cite{chaum2017cmix} is a new mix-net protocol which aims to provide an anoymous communication tool for users at large scales. In contrast with existing mix-net systems, cMix provides significant performance and security upgrades. 
\begin{figure}
	\centering
	\includegraphics[height=8cm, width=9cm]{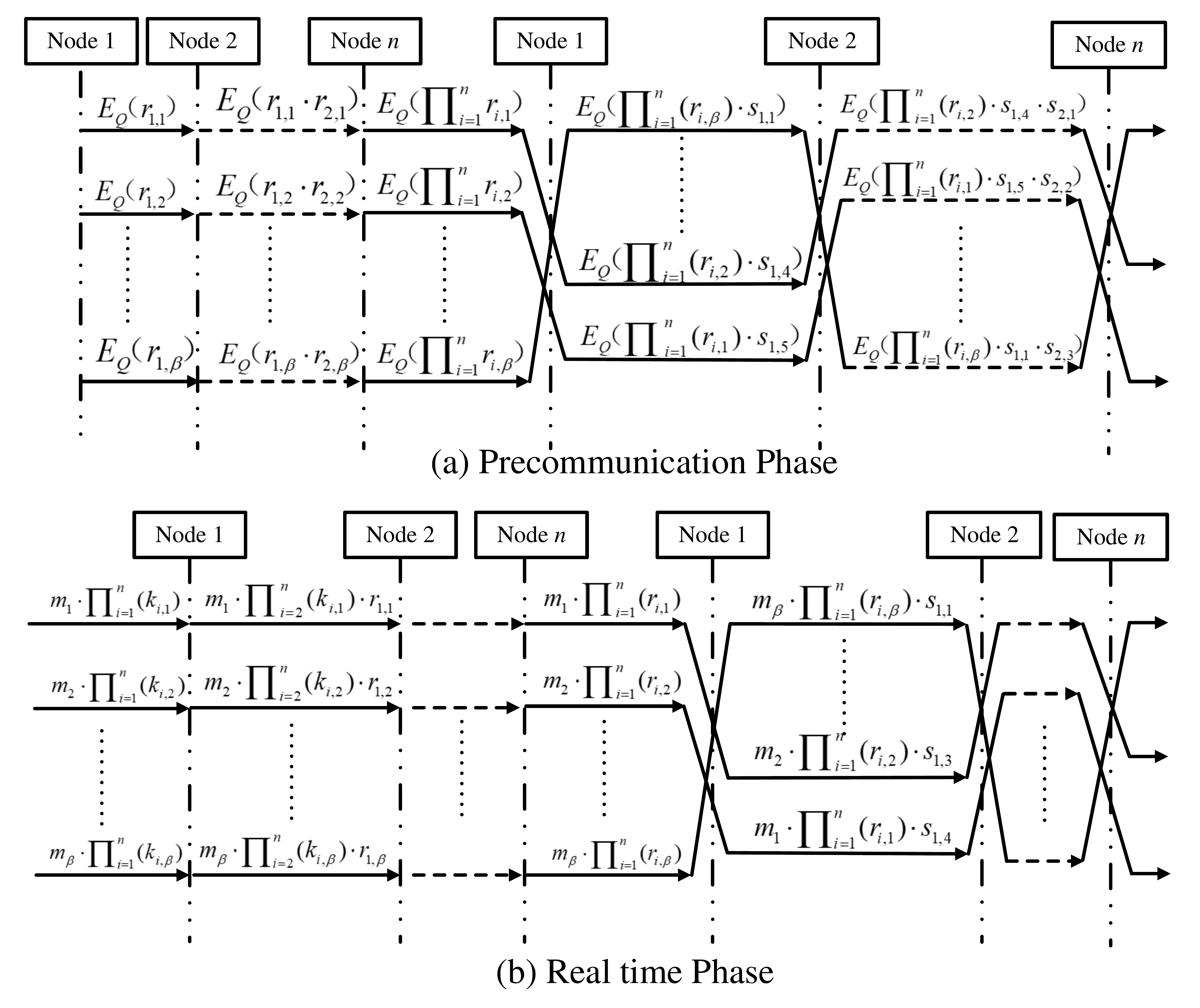}
	\caption{Workflow of cMix.}
	\label{fig1}
\end{figure}
Figure \ref{fig1} briefly describes the workflow of cMix. The protocol contains two participants: $Users:=(U_1,\cdots,U_{\beta})$ and $mix\, nodes:=(N_1,\cdots,N_n)$. Each node $N_i$ holds a tuple of the form $(\pi_i,r_{i,j},s_{i,j},K_{i,t},E(\cdot),D(\cdot))$, where $\pi_i$ is a random permutation, $r_{i,j},s{i,j}\in \mathbb{G} (i\neq j)$ is the random elements in cyclic group $\mathbb{G}$, $K_{i,t}\in \mathbb{G} $ denotes the shared group element between node $N_i$ and user $U_t$, $E(\cdot)$ and $D(\cdot))$ denotes the encryption and decryption algorithm of Elgamal. The detailed description of cMix protocol is provided in Appendix A.

We now present a collision tagging attack on cMix protocol. The attacker have to compromise the last node $N_l$ and any mix node $N_i$. To launch the attack, only small changes are needed to the protocol:
\begin{itemize}
	\item \textbf{Precomputation Phase- Step 3:} The mix node $N_i$ calculate the decryption share $\mathcal{D}(\vec{C}_1)$ with the vector $(1,\cdots,t^{-1},\cdots,1)$ and commit to the $\mathcal{D}(\vec{C}_1)$.
	\item \textbf{Real time Phase- Step 1:} The mix node $N_i$ adds tag $(1,\cdots,t,\cdots,1)$ to $\vec{m}$, and sends $\vec{m}$ to the next mix node. 
	\item \textbf{Real time Phase- Step 3:} The last node publish the output of the mixing step $\Pi_h(\vec{m}\times \vec{R}_h)\times \vec{T}_h$. Afterwards, all mix nodes release their decryption shares $\mathcal{D}_i(\vec{C}_1)$ and the message component of the ciphertext $\vec{C}_2$. The mix node $N_i$ waits for other mix nodes to publish their decryption first, then $N_i$ collude with $N_l$ to obtain the message package with the tagged messages. After that, $N_i$ publish $\mathcal{D}(\vec{C}_1)$.
\end{itemize}

The mix node $N_i$ and $N_l$ get the location of the tag message in advance, and the mixed messages are the same from the senders' perspective. Thus the attacker can break the anonymity of cMix.
\section{Security Model and Requirements}
In this section, we present our proposed system model for BCMIX and the related security requirements. The communication methods among them include transactions and Transport Layer Security (TLS), where the former is an on-chain communication (i.e., publishing a transaction using P2P communications) and the latter is an off-chain communication (i.e., establishing a secure communication channels among mix nodes).
\subsection{System Model}
There are three entities in our proposed BCMIX, that is, Miners, Mix nodes and Senders (see Figure \ref{fig2}). 
\begin{itemize}
	\item \textbf{Miners:} These entities validate new transactions and record them on the global ledger. Simultaneously, the entities compete to solve a difficult mathematical puzzle based on a cryptographic hash algorithm. In BCMIX, miners are eligible to become mix nodes through PoW algorithm competition. Miners disclose their addresses in the form of $(address_M,pk_M)$ where $address_M$ is the blockchain address and $pk_M$ is the blockchain public key deriving the related address.
	\item \textbf{Mix nodes:} These entities are selected from miners through the PoW and VRF algorithm. After being elected as mix nodes successfully, these entities firstly negotiate keys with senders in the set up phase. Then they execute the precomputation and real time phase to encrypt and mix the messages during the duty period and pass the message down. Besides, they should commit to their computations and send special transactions to blockchain network for subsequent auditing.
	\item \textbf{Senders:}  These entities refers to BCMIX users, who hold the respective accounts $address_U,pk_U$. Before accessing to the anonymous service, senders send transactions to mix nodes and negotiate the corresponding keys with mix nodes. Then in the real time phase, senders blind messages with the shared keys and send the message to the first mix nodes.
\end{itemize}
\begin{figure}
	\centering
	\includegraphics[height=5cm, width=9cm]{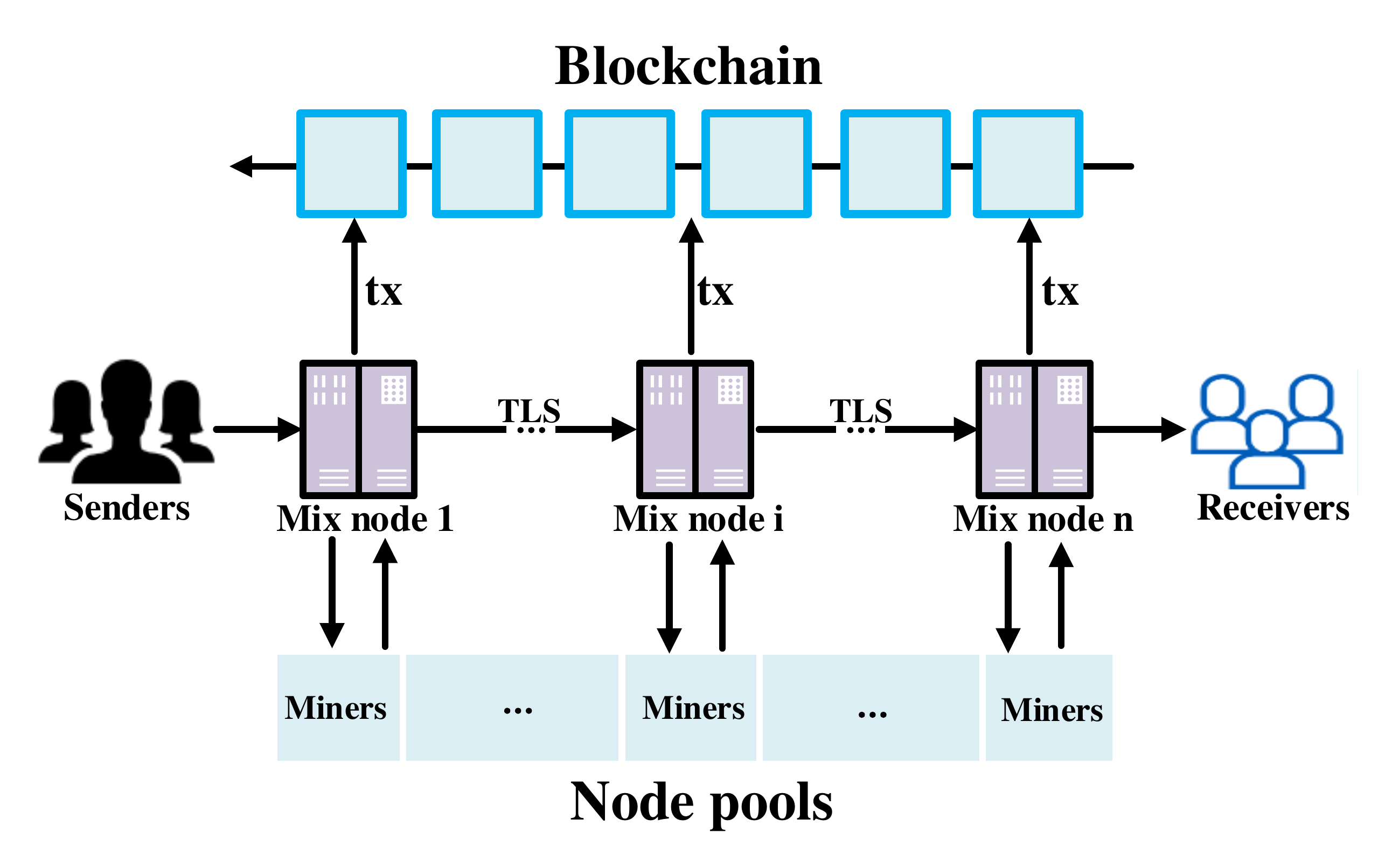}
	\caption{Architecture of blockchain-based mix anonymous system. Before anonymity phase, we elect miners in node pools to serve as mix nodes via designed algorithms. During anonymity phase, mix nodes mix and forward messages from users through Transport Layer Security (TLS). Meanwhile, mix nodes sends special transactions to blockchain for auditing.}
	\label{fig2}
\end{figure}
\subsection{Threat Model}
BCMIX assumes authenticated communication channels among all mix nodes Therefore,  we consider a malicious adversary (a.k.a Byzantine), who can delay, drop, eavesdrop, forward, and delete messages between mix nodes, but not modify, replay, or inject new ones, without detection. For any communication not among mix nodes, we assume the adversary can  delay, drop, re-order, eavesdrop, modify, or inject messages at any point of the network. BCMIX accepts one message per user per batch, starting the pre-computation once the batch reaches $\beta$ messages.

The adversary can also create a lot of accounts and compromise an arbitrary numbers of users. In addition, we assume the adversary can control more than 50\% of the system computing powers. However, such adversary is not able to read the contents of the messages. We assume the security of the used cryptographic primitives, including a secure hash function and a secure signature scheme.
\subsection{Security Requirements}
According to the existing literature \cite{lu2019survey}, \cite{yu2019repucoin}, \cite{hohenberger2014anonize}, BCMIX needs to satisfy the following fundamental security requirements.
\begin{itemize}
	\item {\textbf{Resistance to Sybil Attacks.}} BCMIX should minimize the possibility of an attacker being successfully selected as multiple mix nodes at the same time. 
	\item {\textbf{Resistance to Collision Tagging Attacks.}} BCMIX should prevent collision attackers from performing tagging attacks to link a message to a certain sender.
	\item {\textbf{Sender Anonymity.}} BCMIX provide sender anonymity for users. That is, every mix node performs mixing operations on messages, destroying the associations between senders and receivers. Thus an attacker can not associate teh export messages with a certain sender.
	\item {\textbf{Resistance to MitM attacks.}} BCMIX should prevent an attacker from replacing the shared keys between senders and mix nodes.
	\item {\textbf{Single Point of Failure Resilience.}}  In case of single point of failure (e.g. a mix node is under  denial of service attacks or the mix node is crashed), BCMIX should detect the failure and guarantee the system keep operating.
	\item {\textbf{Resistance to Other Attacks.}} BCMIX should resist common attacks on mix-net system such as replay attacks, traffic-analysis attacks and so on. 
\end{itemize}                                                                                                                                                                                               
\section{Proposed BCMix System}
In this section, we will present our construction of BCMIX system. We first introduce an additive homomorphism mix-net protocol and then we propose the blockchain based mix-net protocol. Thereafter, we describe the concrete BCMIX system.                             
\subsection{Additive Homomorphism Mix-net Protocol}
To solve the MitM attacks of key agreement process and the dependency on trusted entities, we replace Elgamal in cMix \cite{chaum2017cmix} with EC-Elgamal and propose an additive homomorphism mix-net protocol $\Delta$ to integrate the mix-net protocol with the blockchain.

Our mix-net protocol contains two participants, $Senders:=(U_1,\cdots,U_{\beta})$ and $mix\, nodes:=(N_1,\cdots,N_n)$, where the mix nodes are selected from the blockchain miners. Mix nodes negotiate keys with senders by means of a special transaction $tx_{\mathit{KE}}$ and calculate the system public key through their address pair $(address_N^i,pk_N^i)$, where $pk_M^i=Q_i$. (We will describe these two processes in the next part). We suppose a authenticated communication among mix nodes and we denote it as $\Delta.\mathsf{TLS}$. The proposed mix-net protocol is a tuple of algorithms $(\mathsf{Setup},\mathsf{Precom},\mathsf{RealTime})$. The notations and the processes of the protocol $\Delta$ are presented in Table I and Algorithm \ref{alg1}.
\begin {table}[htbp]
\centering\caption{Notations of The Revised Mix-net Protocol}
	\footnotesize
	\label{tab1}
\begin{tabular}{l p{7cm}}
	\Xhline{0.7pt}
	\textbf{Symbol} &            \textbf{Description}  \\
	\Xhline{0.7pt}
	$d_i$& the secret share for mix node $N_i$ of the secret key $d$, $d_i\in\mathbb{F}_p$;\\
	\Xhline{0.7pt}
	$Q_i$& the public key of mix node $N_i$, $Q_i=d_iG$;\\
	\Xhline{0.7pt}
	$\boldsymbol{Q}$& the public key of the system, $\boldsymbol{Q}=\sum_iQ_i$;\\
	\Xhline{0.7pt}
	$\mathcal{E}_Q()$& $\mathcal{E}_Q(m)=(x\cdot G,m+x\cdot \boldsymbol{Q}),x\in\mathbb{F}_p$;\\
	\Xhline{0.7pt}
	$\mathcal{D}_{d_i}()$& the decryption share of mix node $N_i$, $\mathcal{D}_{d_i}(\sum_i{x}_i\cdot G)=(\sum_i{x}_i)d_i\cdot G$;\\
	\Xhline{0.7pt}
	$r_{i,a}$ & random values (freshly generated for each round) of mix node $N_i$ for groove $a$;  \\
	\Xhline{0.7pt}
	$s_{i,a}$ & random values (freshly generated for each round) of mix node $N_i$ for groove $a$; \\
	\Xhline{0.7pt}
	$\pi_i$ & a random permutation of the $\beta$ grooves used by mix node $N_i$; \\
	\Xhline{0.7pt}
	$\boldsymbol{\Pi}_i$ & the permutation performed by BCMix through mix node $N_i$; \\
	\Xhline{0.7pt}
	$k_{i,j}$    & a group element shared between mix node $N_i$ and the sending user for groove $j$. These values are used as keys to blind messages;  \\  
	\Xhline{0.7pt}
	$\boldsymbol{k}_{i}$    &  the vector of derived secret keys shared between mix node $N_i$ and all users in a batch, i.e. $\boldsymbol{k}_{i}=(k_{i,1},\cdots,k_{i,\beta})$; \\
	\Xhline{0.7pt}
	$K_{j}$    &  the product of all shared keys for the sending user of slot $j$, i.e. $K_{j}=\prod_{i=1}^nk_{i,j}$ ; \\
	\Xhline{0.7pt}
	$M_j $ & the message sent by user $j$. Like other values in the system, these vlaues are group elements;\\
	\Xhline{0.7pt}
	$\boldsymbol{R}_i$,     & the product of all local random $\boldsymbol{r}$ values through mix node $N_i$, $\boldsymbol{R}_i=\prod_{j=1}^i\boldsymbol{r}_j$;\\
	\Xhline{0.7pt}
	 $\boldsymbol{S}_i$&the product of all local random $\boldsymbol{s}$ values through mix node $N_i$,  $\boldsymbol{S}_i=\begin{cases}
	 \boldsymbol{s}_1 & i=1\\
	 \pi_i(\boldsymbol{S}_{i-1})\times  \boldsymbol{s}_i & 1<i\leq n  
	 \end{cases}$.\\
	 \Xhline{0.7pt}
\end{tabular}
\end{table}
\subsection{Basic Components of the Proposed Blockchain Protocol}
We build our blockchain protocol $\Gamma^*$ by extending and modifying the aforementioned protocol $\Gamma$. We first define the basic components in our blockchain protocol. 

\textbf{Transaction.} Our blockchain contains three types of transaction, namely normal transaction $tx_{\mathit{N}}$, key-exchange transaction $tx_{\mathit{KE}}$ and commitment transaction $tx_{\mathit{COM}}$. The normal transaction $tx_{\mathit{N}}$ is the same definition of the transaction that in protocol $\Gamma$. We define a key-exchange transaction $tx_{\mathit{KE}}:=(\mathit{KE},pk,\overrightarrow{inputs},\overrightarrow{outputs},sig)$ and a commitment transaction $tx_{\mathit{COM\mathit}}:=(\mathit{COM},commitment,sig)$, where $\mathit{KE}$, $\mathit{COM}$ denotes the transaction type, $pk$ is the blockchain public key of the sender, $commitment$ is the commitment value issued by mix nodes and $\overrightarrow{inputs}$, $\overrightarrow{outputs}$, $sig$ are the same as the above definition. A key-exchange transaction is used to negotiate keys between senders and mix nodes, while a commitment transaction is released to supervise the behavior of mix nodes.

\textbf{Block and Chain.} According to the transaction type involved, we define three blocks respectively called main block $B_{\mathit{M}}$, key-exchange block $B_{\mathit{KE}}$ and commitment block $B_{\mathit{COM}}$, where $B_{\mathit{M}}:=(s_{\mathit{ke}},s_{\mathit{com}},x,tx_{\mathit{N}}s,ctr)$, $B_{\mathit{KE}}:=(s,x,tx_{\mathit{KE}}s)$ and  $B_{\mathit{COM}}:=(s,x,tx_{\mathit{COM}}s)$. Here $s$ is the state of the previous block, $s_{\mathit{ke}}$ is the state of previous  blocks, $x$ is the data and ctr is the proof of work of the block. 

A chain $\mathcal{C}$ is the form of $\mathcal{C}:=(\mathcal{C}_{\mathit{M}},\mathcal{C}_{\mathit{KE}},\mathcal{C}_{\mathit{COM}})$, where $\mathcal{C}_{\mathit{M}}$, $\mathcal{C}_{\mathit{KE}}$ and $\mathcal{C}_{\mathit{COM}}$ respectively represent main chain, key-exchange chain and commitment chain in Figure \ref{fig3}. We define a stable main block $B_{\mathit{M}}^i$ as the origin of the key-exchange chain $\mathcal{C}_{\mathit{KE}}$ and commitment chain $\mathcal{C}_{\mathit{COM}}$.
\begin{figure}
	\centering
	\includegraphics[height=5cm, width=8cm]{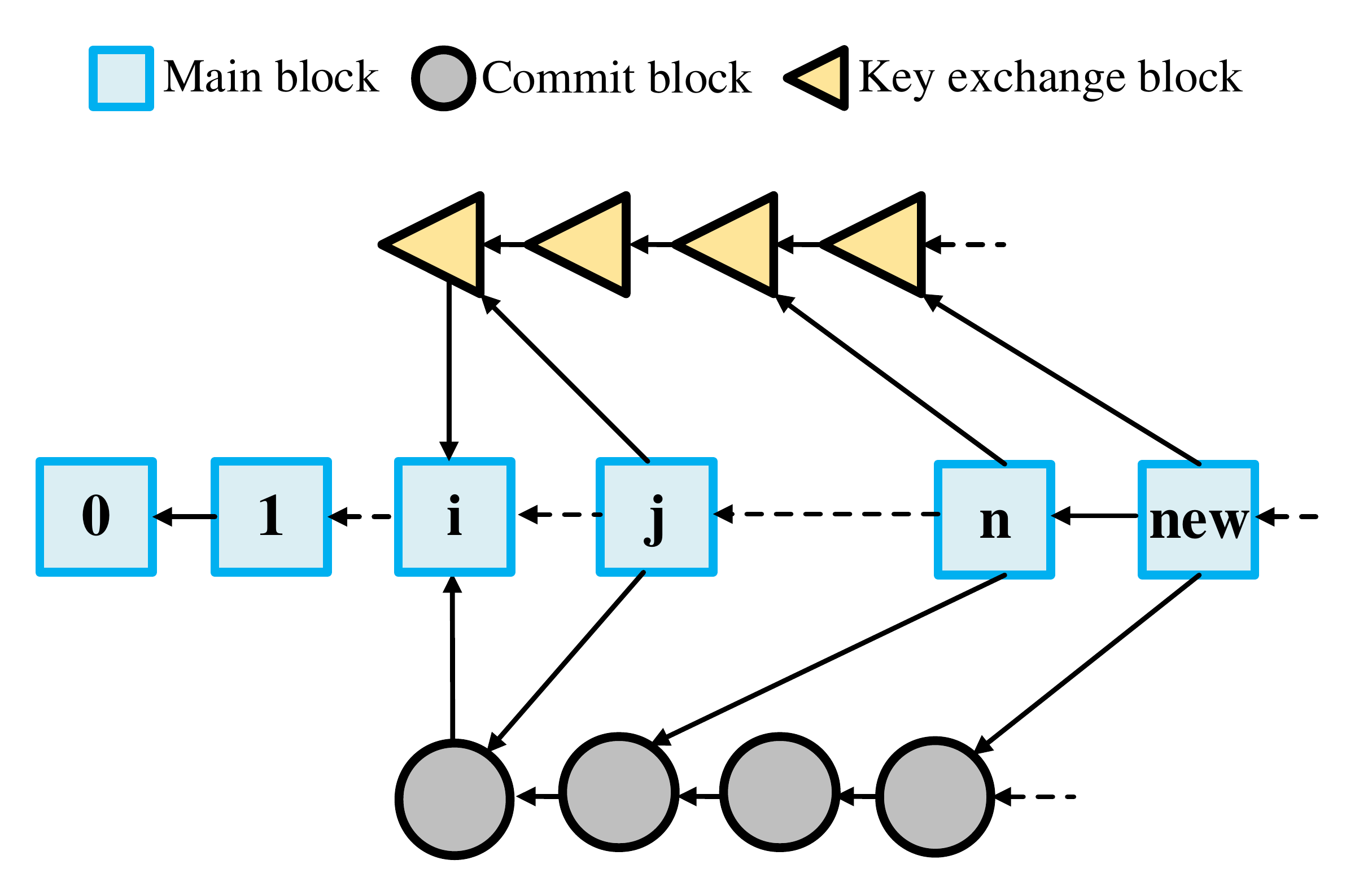}
	\caption{Chain Structure of BCMIX}
	\label{fig3}
\end{figure}
\begin{definition}
	Let two chains $\mathcal{C}_1$, $\mathcal{C}_2$ are possessed by two honest parties at the onset of the slots $slot_1<slot_2$, if $\mathcal{C}_1^{\ulcorner w}\preceq \mathcal{C}_1$ and $\mathcal{C}_1^{\ulcorner (w+1)}\npreceq \mathcal{C}_1$, then we call $\mathcal{C}_{\mathit{M}}^w$ is a stable main chain and $B_{\mathit{M}}^i$ where $i\in[1,w]$ are stable main blocks.
\end{definition}	
Note that the $B_{\mathit{KE}}$ and the $B_{\mathit{COM}}$ do not contain proof of work, thus an adversary can manipulate the contents of these blocks. To avoid the problem we specify that the latest main block $B_{\mathit{M}}^{latest}$ always contains the states of $B_{\mathit{KE}}^{latest}$ and $B_{\mathit{COM}}^{latest}$. We further give out the definition of valid blocks.
\begin{definition}
	We say blocks $B_{\mathit{N}}$, $B_{\mathit{KE}}$ and $B_{\mathit{COM}}$ are valid iff 
	\begin{itemize}
		\item  The transactions contained in $B_{\mathit{N}}$, $B_{\mathit{KE}}$ and $B_{\mathit{COM}}$ are valid;
		\item  For any two main blocks $B_{\mathit{N}}^i:=(s_{\mathit{ke}}^{p},s_{\mathit{com}}^{q},x^i,tx_{\mathit{N}}^is,ctr^i)$ and $B_{\mathit{N}}^j:=(s_{\mathit{ke}}^{m},s_{\mathit{com}}^{n},x^j,tx_{\mathit{N}}^js,ctr^j)$, where $i<j$. $p\ge m$ and $q \ge n$ hold;
		\item  $ H(ctr,T(s_{\mathit{ke}},s_{\mathit{com}},x,txs))<D$.
		\end{itemize}	
\end{definition}
\begin{algorithm*}[htbp]
	\caption{\centerline {The Additive Homomorphism Mix-net Protocol $\Delta$}}
	\footnotesize
	\label{alg1}
	\begin{algorithmic}[1]
		\Procedure{\textbf{Setup Phase}}{}\\
		$(Q,K_i,k_{i,j})\leftarrow\mathsf{Setup}(pk_U^i,pk_M^j)$: This algorithm is invoked by senders $U_i$ with address $(address_U^i,pk_U^i)$ and mix nodes $M_j$ with address $(address_N^j,pk_N^j)$, where $i\in[1,\beta]$,  $j\in[1,n]$. The algorithm takes as inputs the public key of senders and mix nodes and returns a shared key $k_{i,j}$ between a send $U_i$ and a mix node $M_j$. A slot key $K_i=\prod_{j=1}^{n}k_{i,j}$ and the system public key $Q=\sum_jpk_M^j=\sum_jQ_j$.
		\EndProcedure
		\Procedure{\textbf{Precomputation Phase}}{}\quad$\mathsf{Precom}:=(\mathsf{Preprocess},\mathsf{Mix},\mathsf{Postprocess})$
		\begin{itemize}
			\item $\mathsf{Preprocess}$. The mix nodes generate the fresh $\boldsymbol{r},\boldsymbol{s},\pi$ values and computes the encryption $\mathcal{E}(\boldsymbol{r}_i^{-1})$. At the same time the mix nodes issues the commitment values of the fresh values $COM_{r},COM_{s},COM_{\pi}$. Then they collectively compute the product of the received values by sending the following message to the next mix node:
			$$\mathcal{E}_{\boldsymbol{Q}}(\boldsymbol{R}_i)=\begin{cases}
			\mathcal{E}_{\boldsymbol{Q}}(\boldsymbol{r}_1) &  i=1 \\
			\mathcal{E}_{\boldsymbol{Q}}(\boldsymbol{R}_{i-1})\times \mathcal{E}_{\boldsymbol{Q}}(\boldsymbol{r}_i)  &  1<i \leq n
			\end{cases}$$
			Eventually, the last mix node sends the final values $\mathcal{E}_{\boldsymbol{Q}}(\boldsymbol{R}_n)$ to the first mix node as input for the next step and issues the commitment value $COM_{\mathcal{E}_{\boldsymbol{Q}}(\boldsymbol{R}_n)}$.
			\item $\mathsf{Mix}$. The mix node together mix the values and compute the results $\boldsymbol{\Pi}_n(\boldsymbol{R}_n)\times\boldsymbol{S}_n$, under encryption. The mix nodes perform this mixing by having each mix node $i$ send the following message to the next mix node:
			$$\mathcal{E}_{\boldsymbol{Q}}(\boldsymbol{\Pi}(\boldsymbol{R}_n)\times\boldsymbol{S}_i)=\begin{cases}
			\pi_1(\mathcal{E}_{\boldsymbol{Q}}(\boldsymbol{R}_n)\times\mathcal{E}_{\boldsymbol{Q}}(\boldsymbol{s}_1))&  i=1\\
			\pi_i(\mathcal{E}_{\boldsymbol{Q}}(\boldsymbol{\Pi}_{i-1}(\boldsymbol{R}_n)\times \boldsymbol{S}_{i-1}))\times\mathcal{E}_{\boldsymbol{Q}}(\boldsymbol{s}_i)  &  1<i \leq n
			\end{cases}$$
			As with the first step, the last mix node sends the final encrypted values $\mathcal{E}_{\boldsymbol{Q}}(\boldsymbol{\Pi}_n(\boldsymbol{R}_n)\times\boldsymbol{S}_n)$ to the first mix node.
			\item $\mathsf{Postprocess}$. To complete the precomputation, each mix node $N_i$ computes its decryption shares $\mathcal{D}_{d_i}(\sum_i{x}_i\cdot G)$, where $(x,c)=\mathcal{E}_{\boldsymbol{Q}}(\boldsymbol{\Pi}_n(\boldsymbol{R}_n)\times\boldsymbol{S}_n)$, and keep its secret. Then each mix node issues the commitment values $COM_{\mathcal{D}_{d_i}}$ of their secret shares. The message parts $c$ are multiplied with all the decryption shares to retrieve the plaintext values $\boldsymbol{\Pi}_n(\boldsymbol{R}_n)\times\boldsymbol{S}_n$. The last mix node to be used in the real time phase stores the decrypted precomputed values.
		\end{itemize}
		\EndProcedure
		\Procedure{\textbf{RealTime Phase}}{}\quad$\mathsf{RealTime}:=(\mathsf{Preprocess},\mathsf{Mix},\mathsf{Postprocess})$\\
		Each user constructs its message $MK_j^{-1}$ (for slot $j$) by multiplying the message $M_j$ and it sends it to the first mix node, which collects all messages and combines them to get a vector $\boldsymbol{M}\times \boldsymbol{K}^{-1}$.
		\begin{itemize}
			\item $\mathsf{Preprocess}$. Each node $N_i$ sends $\boldsymbol{k}_i \times \boldsymbol{r}_i$ to the next mix node, which uses them to compute $\boldsymbol{M}\times \boldsymbol{R}_n=\boldsymbol{M}\times\boldsymbol{K}^{-1}\times \boldsymbol{\prod}_{i=1}^n\boldsymbol{k}_i\times \boldsymbol{r}_i$ and the last node sends the result to the first node.
			\item $\mathsf{Mix}$. Each node $N_i$ computes  $\pi_i(\boldsymbol{\Pi}_{i-1}(\boldsymbol{M}\times\boldsymbol{R}_n)\times \boldsymbol{S}_{i-1})\times\boldsymbol{s}_i$, where $\boldsymbol{\Pi}_0$ is the identity permutation and $\boldsymbol{S}_0=1$. The last node sends a commitment to its message $\boldsymbol{\Pi}_n(M\times \boldsymbol{R}_n)\times \boldsymbol{S}_n$ to every other node.
			\item  $\mathsf{Postprocess}$. Each node $N_i$ opens its precomputed decryption share for $(\boldsymbol{x},\boldsymbol{c})=\mathcal((\boldsymbol{\Pi}_n(\boldsymbol{R}_n)\times \boldsymbol{S}_n)^{-1})$, while the last node sends its decryption share multiplied by the value in the previous step and the message component: $\boldsymbol{\Pi}_n(\boldsymbol{M}\times\boldsymbol{R}_n)\times \boldsymbol{S}_n\times\mathcal{D}_n(\boldsymbol{x})\times \boldsymbol{c}$. Finally, the permuted message can be decrypt as $\boldsymbol{\Pi}_n(\boldsymbol{M})=\boldsymbol{\Pi}_n(\boldsymbol{M}\times\boldsymbol{R}_n)\times\boldsymbol{S}_n\times\prod_{i=1}^{n}\mathcal{D}_i(\boldsymbol{x})\times\boldsymbol{c}$.
		\end{itemize}
		\EndProcedure
		\rule[0pt]{17.65cm}{0.063em}\\
		\begin{scriptsize}
			\textbf{Note:} The shared key $k_{ij}$ in \textbf{Setup Phase} is generated through ECDH protocol. After the mix node set $(N_1,N_2,\cdots,N_n)$ is selected from miners. Every sender sends $tx_{\mathit{KE}}$s to all mix nodes.
			Since senders and mix nodes know the public key of each other, they can execute the ECDH protocol and obtain the shared key $k_{ij}$.
		\end{scriptsize}
	\end{algorithmic}
\end{algorithm*}
\textbf{Blockchain based mix-net protocol.} Below we propose a blockchain based mix-net protocol $\Gamma^*$ which combines the basic blockchain protocol $\Gamma$ with the proposed mix-net protocol $\Delta$. The protocol $\Gamma$ has copies of all the basic functionalities exposed by $\Gamma$ and $\Delta$ through the interfaces described above, and adds additional algorithms including VRF, IP Sharding in order to resist Sybil attacks. We describe the protocol $\Gamma^*:=(\mathsf{Setup},\mathsf{Update},\mathsf{Vote},\mathsf{Mix},\mathsf{Verify},\mathsf{Broadcast})$ as follows.
\begin{itemize}
	\item $\mathsf{Setup}$. System parameters involved in our construction are $\{E_p(a,b),\mathbb{G},G,\mathcal{K},H,T\}$, where $E_p(a,b)$ is a non-singular elliptic curve, $\mathbb{G}$ is a cyclic group which consists of all points on $E_p(a,b)$, as well as the point at infinity $\mathcal{O}$, $G$ is the base point of $E_p(a,b)$, $\mathcal{K}$ is the security parameter and $H:\{0,1\}^*\rightarrow\{0,1\}^{\mathcal{K}}$, $T:\{0,1\}^*\rightarrow\{0,1\}^{\mathcal{K}}$ are cryptographic hash functions. Entities invoke the algorithm to generate the blockchain key pair $(pk,sk):=(qG,q)$, the address $(address,pk)$ and the VRF key pair $(vk,vsk)$. Here $q\in \mathbb{F}_p$ is selected by entities.
	\item $\mathsf{Update}$. This algorithm first invoke $\Gamma.\mathsf{validate}$ to validate the new transactions, chains and blocks. Then returns a longer valid chain $\mathcal{C}:=(\mathcal{C}_M,\mathcal{C}_{KE},\mathcal{C}_{COM})$ in the network (if it exists), otherwise returns $\bot$.
	\item $\mathsf{PoWVote}$. The algorithm take as inputs the data $x$, the state $s$ and a difficulty level $D'<D$, where D is the system difficulty level, and then computes $H(ctr',T(s,x))$ where $ctr'$ is a random string.  Thereafter, the algorithm  estimate whether $H_i(ctr',T(s,x))<D'$. If so the algorithm outputs $1$ and add the miner into a set $\mathcal{M}$, otherwise returns $0$.
	\item $\mathsf{IPSharding}$. This algorithm takes as input IP prefix of miners in set $\mathcal{M}$, and outputs $n$ node pools (Miners in each node pool have the same IP prefix $IP/z$ where $z$ is a integer and $z\in[0,32]$.)
	\item $\mathsf{VRF}$. The algorithm takes as inputs $IP/z$ and the current slot $slot_l$, and outputs a value $Y\in\mathcal{Y}$, where $\mathcal{Y}$ is a finite set, and a proof $\pi$.  
	\item $\mathsf{Mix}$. The algorithm take as inputs a sender set $\mathcal{U}:=(U_1,U_2,\cdots,U_{\beta})$, a mix node set $\mathcal{N}:=(N_1,N_2,\cdots,N_n)$ and a message vector $\boldsymbol{M}$. Then it revokes $(\Delta.\mathsf{Setup},\Delta.\mathsf{Precom},\Delta.\mathsf{RealTime})$ and outputs a permuted message vector $\boldsymbol{\Pi}_n(\boldsymbol{M})$.
	\item $\mathsf{Verify}$. Given a message vector $\boldsymbol{M}:=(m_1,m_2,\cdots,m_{\beta})$ and a message vector $\boldsymbol{M}:=(m_{(1)},m_{(2)}),\cdots,m_{(\beta)})$, this algorithm decides whether the following conditions holds: For $\forall i\in[1,\beta]$, $m_i\in \boldsymbol{M}$, $\exists j\in[1,\beta]$, $m_j\in\boldsymbol{M}$, $m_i=m_j$. If so, the algorithm returns 1. Otherwise the algorithm returns 0.
	\item $\mathsf{Broadcast}$: The algorithm takes as inputs some $tx$s and address $(address,pk)$ and broadcasts them to all the nodes of the blockchain system.
\end{itemize}
\begin{figure*}
	\centering
	\includegraphics[height=8cm, width=18cm]{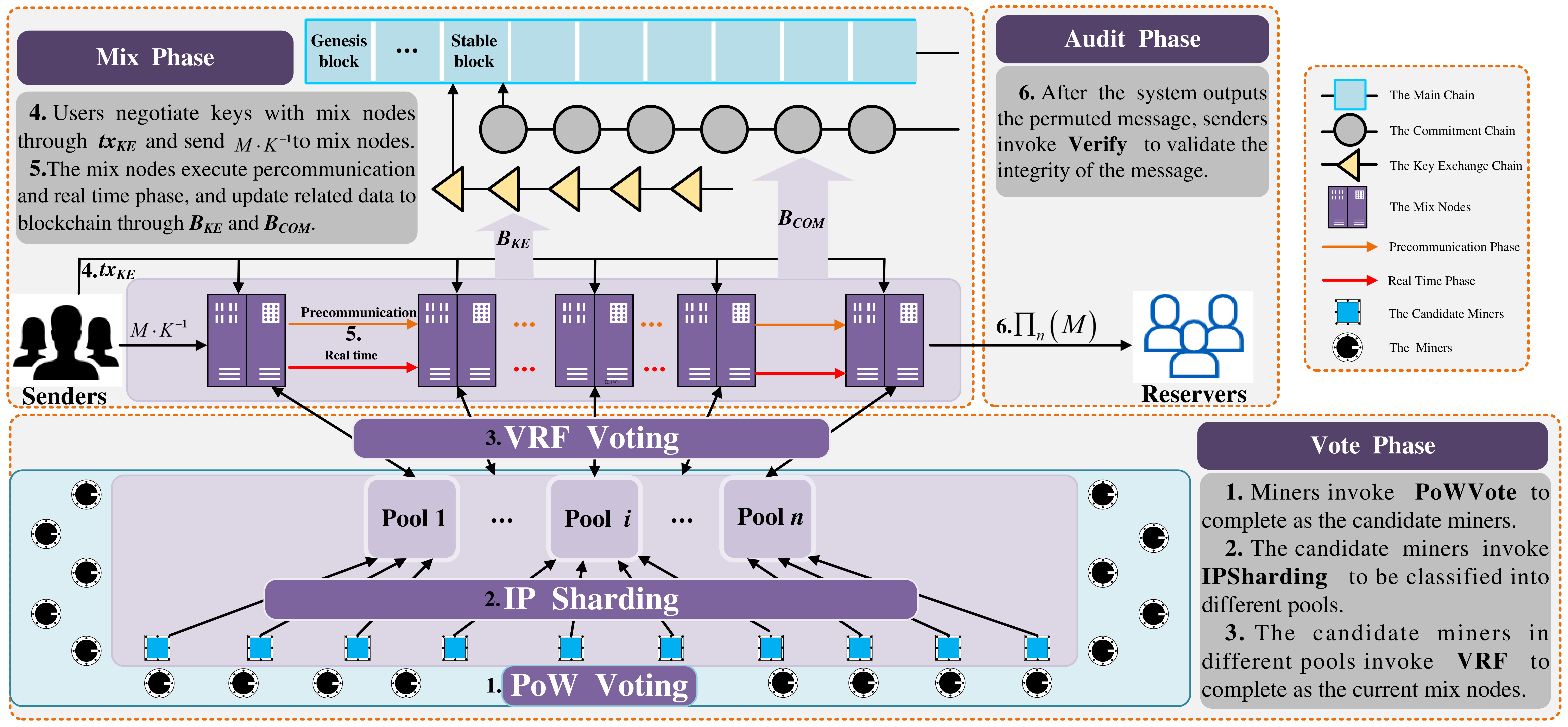}
	\caption{The detailed orchestration of BCMIX.}
	\label{fig4}
\end{figure*}
\begin{theorem}
 If $\Gamma$ satisfies $(\tau,s)$-chain growth, then $\Gamma^*$ satisfies $(\tau,s)$-chain growth. 
\end{theorem}
\begin{proof}
	We note that the side chain $\mathcal{C}_{KE}$ and $\mathcal{C}_{COM}$ do not contain proof of work, and the update of $\mathcal{C}_{KE}$, $\mathcal{C}_{COM}$ is independent of the update of $\mathcal{C}_{M}$. Thus we conclude that $\mathcal{C}_{M}$ satisfies $(\tau,s)$-chain growth as the chain $\mathcal{C}$ in $\Gamma$. We now prove that $\mathcal{C}_{KE}$ and $\mathcal{C}_{COM}$ satisfy $(\tau,s)$-chain growth.
	
	Consider the chains $\mathcal{C}_{KE}^1,\mathcal{C}_{KE}^2$ and $\mathcal{C}_{COM}^1,\mathcal{C}_{COM}^2$ at the onset of two slots $slot_1,slot_2$, with $slot_2$ at least $s$ slots ahead of $slot_1$, and $\tau_{KE},\tau_{COM}$ are the speed coefficient of $\mathcal{C}_{KE}$ and $\mathcal{C}_{COM}$. Since a key exchange block $B_{KE}$ and a commitment block $B_{COM}$ are jointly generated by all mix nodes at the expected speed $\tau_{KE}$ and $\tau_{COM}$ respectively, we derive that $\mathcal{C}_{KE}$ and $\mathcal{C}_{COM}$ are updated at the expected speed $\tau_{KE}$ and $\tau_{COM}$. Then we hold the following two equations $\mathsf{len}(\mathcal{C}_{KE}^2)-\mathsf{len}(\mathcal{C}_{KE}^1)=\tau_{KE}\cdot s$ and $\mathsf{len}(\mathcal{C}_{COM}^2)-\mathsf{len}(\mathcal{C}_{COM}^1)=\tau_{COM}\cdot s$. Thus we conclude that $\mathcal{C}_{KE}$ and $\mathcal{C}_{COM}$ satisfies $(\tau,s)$-chain growth. In summary, $\Gamma^*$ satisfies $(\tau,s)$-chain growth. 
\end{proof}
\begin{theorem}
	Let $H,G$ be two collision-resistant cryptographic hash functions. If $\Gamma$ satisfies $(\mu,\ell)$-chain quality, then $\Gamma^*$ satisfies $(\mu,\ell)$-chain quality. 
\end{theorem}
	\begin{proof}
		We emphasize that proof of work is not contained in the key exchange chain $\mathcal{C}_{KE}$ and the commitment chain $\mathcal{C}_{COM}$. And the security of $\mathcal{C}_{KE}$, $\mathcal{C}_{COM}$ is depend on the security of the main chain $\mathcal{C}_M$ since $\mathcal{C}_M$ takes the states of $\mathcal{C}_{KE}$ and  $\mathcal{C}_{COM}$ as inputs of the proof of work. Suppose an adversary $\mathcal{A}$ wants to manipulate contents of $\mathcal{C}_{KE}$ and $\mathcal{C}_{COM}$. According to the formula $H(ctr,T(s_{ke},s_{com},x,txs))<D$ (Definition 7.), $\mathcal{A}$ has to amend the corresponding main block $B_M$. Note that the capabilities of the adversary and external environment in $\Gamma^*$ are exactly the same as that in $\Gamma$. Thus we can conclude that the main chain $\mathcal{C}_M$ is the only factor affecting the chain quality property. We show below that $\mathcal{A}$ has only a negligible probability of violating chain quality of $\Gamma^*$.
		
		Let us denote by $B_M^i$ the $i$-th block of the main chain $\mathcal{C}_M$ at some slot intervals so that $\mathcal{C}_M:=B_M^1\cdots B_M^{\mathsf{len}(\mathcal{C})}$. From Definition 4. we know that the number of main blocks generated by $\mathcal{A}$ in chain $\mathcal{C}_M$ are at most $\mu\cdot \mathsf{len}(\mathcal{C})$. According to \cite{garay2015bitcoin}, $\mathcal{A}$ can not generate more than $\mu\cdot \mathsf{len}(\mathcal{C})$ main blocks with the current computing hash power. Or the adversary $\mathcal{A}$ could try to build an valid candidate block $B_M^{\star}$ to replace a validate block $B_M^j$ generated by honest parties in $\mathcal{C}_M$, where $j\in[1,\mathsf{len}(\mathcal{C})]$, $B_M^{\star}\neq B_M^j$ and $H(B_M^{\star})=H(B_M^j)$. By the collision-resistance property of hash function we can draw a conclusion that the adversary has only a negligible chance of producing such a candidate block $B_M^{\star}$ where $H(B_M^{\star})=H(B_M^j)$. Hence $\Gamma^*$ satisfies $(\mu,\ell)$-chain quality. 
	\end{proof}

\begin{theorem}
	Let $B_M^w$ be a stable block which is the genesis of the chain $\mathcal{C}_{KE}$ and $\mathcal{C}_{COM}$. If $\Gamma$ satisfies $k$-common prefix, then $\Gamma^*$ satisfies $k$-common prefix.
\end{theorem}
\begin{proof}
	Note that the chain $\mathcal{C}_{KE}$ and $\mathcal{C}_{COM}$ will not fork since they don not contain proof-of-work and are uniquely generated bt the mix nodes. Hence $\mathcal{C}_{KE}$ and $\mathcal{C}_{COM}$ satisfy the $k$-common prefix property. We recall that the behaviors of the adversary and the honest parties in $\Gamma^*$ are exactly the same as that in $\Gamma$. Thus according to the literature \cite{}, the main chain satisfies $k$-common prefix property. This concludes the proof.
\end{proof}

Given the above, the tuple $\Gamma^*:=(\mathsf{Setup},\mathsf{Update},\mathsf{Vote},\mathsf{Mix},\mathsf{Verify})$ is a secure blockchain protocol which satisfies the properties of \textit{$(\tau,s)$-chain growth}, \textit{$(\mu,\ell)$-chain quality} and \textit{$k$-common prefix}.
\subsection{System Design} 

With the description of the blockchain-based mix-net protocol, we propose the dynamic self-organizing blockchain-based mix anonymous system in this part. Our BCMIX system consists of four phases, namely: \textbf{System Initialization}, \textbf{Vote}, \textbf{Mix}, \textbf{Audit}. The orchestration of BCMIX is detailed in Figure \ref{fig4}.

\textbf{System Initialization}. This phase initializes the system parameters and generate accounts for participants. The system determine the public system parameters $\{E_p(a,b),\mathbb{G},G,\mathcal{K},H,T\}$ (e.g. $\mathsf{secp256k1}$ of Bitcoin). After that the entities (i.e., users and miners) invoke $\Gamma^*\mathsf{Setup}$ to generate their key pairs (i.e., user key pair $(pk_u^i,sk_u^i):=(q_u^iG,q_u^i)$ and miner key pair $(pk_m^j,sk_m^j):=(q_m^jG,q_m^j)$) and addresses (i.e., user address $(address_U^i,pk_U^i)$ and miner address $(address_N^j,pk_N^j)$ respectively). In additon,  miners also generate their VRF key pairs, which are used to complete as mix nodes.

\textbf{Vote}. This phase is executed by the miners to elect the candidate mix nodes for mixing messages.
\begin{itemize}
	\item [1)] Firstly, miners invokes $\Gamma^*.\mathsf{PoWVote}$ to solve the puzzle $D'$. If the algorithm $\Gamma^*.\mathsf{PoWVote}$ returns 1, then the miner $\mathcal{M}_i$ is added to the candidate set $\mathcal{M}$. Otherwise miners re-select inputs $s,x$ to compute $H(ctr',T(s,x))$ until $H(ctr',T(s,x))<D'$.
	\item [2)] After a period of certain time, e.g. $T$, $\mathcal{M}$ stops accepting new miners. Then the candidates miners run $\Gamma^*.\mathsf{IPSharding}$ (Algorithm \ref{alg2}) and join the node pool with the same IP prefix.
	\begin{algorithm}
		\caption{  \centerline{The IPSharding algorithm.}}
		\label{alg2}
		\begin{algorithmic}[1]
			\Require
			The IP address of the candidate miners, where $ip_i:=A_1.A_2.*.*$. Here $A_1$, $A_2$ and $*$ denote four decimal segments; the system parameter $\sigma$;    
			\Ensure
			 $k$ mix node pools;
	     	\State  $coordinate_{ip_m}:=(A_1^m,A_2^m)$;
	     	\State  $|Num_{2^p,2^q}|$ denotes the number of nodes distributed in $(2^p,2^q)$;
	     	\For {$j=0$,$j<=2$}
		    	\For {$i=0$,$i<=7$}
		        	\If {$A_j^m<2^i$ or $A_j^m>2^{i+1}$}
		            	\State $i++$;
		        	\Else\quad {break};
		        	\EndIf
		         	\State $j++$;
               \EndFor
            \EndFor
            \State Classify the coordinate into $k$ parts such that $\frac{|Num_{2^p,2^q}|_j}{|Num_{2^p,2^q}|_{min}}<=\sigma$, $k\in[1,k]$;\\
			\Return $k$ mix node pools.
		\end{algorithmic}
	\end{algorithm}
	\item [3)] After that the candidate miners $\mathcal{M}_i$ in each node pool take the $IP/z$ and current slot $slot_l$ as inputs, and invoke $\Gamma^*.\mathsf{VRF}$ to generate a value $Y_i\in\mathcal{Y}$ and the corresponding proof $\pi$. Then the candidate miner who holds the smallest value $Y$ in each node pool is elected as the mix node in the current slot. Finally, the elected mix nodes are networked in the order of $Y$.
\end{itemize}

\textbf{Mix}. In this phase, the current mix nodes $\mathcal{N}:=(N_1,N_2,\cdots,N_n)$ first negotiate shared keys with users $\mathcal{U}:=(U_1,U_2,\cdots,U_{\beta})$. Then the current mix nodes blind and mix messages $\boldsymbol{M}:=(\boldsymbol{M}_1,\boldsymbol{M}_2,\cdots,\boldsymbol{M}_{\beta})$ for users.

To negotiate keys with each other through $\mathsf{ECDH}$ protocol, the mix nodes and the users do as follows. 
\begin{itemize}
	\item [1)] After the current mix nodes are selected, they run $\Gamma^*.\mathsf{Broadcast}$ to disseminate their addresses and the VRF parameters, i.e. $(address,pk)||(vk,Y,\pi)$, to the participants in the blockchain network.
	\item [2)] Upon receiving the addresses and parameters of all mix nodes, a user $U_i$ first invoke $\mathsf{VRF}.\mathsf{Ver}$ to verify the received. If the algorithm returns 1, then the user $U_i$ sends transactions $tx_{KE}^i$ to every mix nodes. Then the mix nodes parse $tx_{KE}^i$ as $tx_{\mathit{KE}}:=(\mathit{KE},pk,\overrightarrow{inputs},\overrightarrow{outputs},sig)$ and obtain the public key of each user.
	\item [3)] Finally the mix nodes $N_j$ and the user $U_i$ multiply their secret key with the public key of each other, and obtain the shared key $k_{ij}=q_U^i\cdot q_N^j\cdot G$.
\end{itemize}

After negotiating keys with users, the mix nodes network blind and mix messages as follows.
\begin{itemize}
	\item [1)] The mix nodes invoke $\Delta.\mathsf{Precom}$ to compute the parameters $\mathcal{E}_{\boldsymbol{Q}}(\boldsymbol{\Pi}_n(\boldsymbol{R}_n)\times\boldsymbol{S}_n)$ utilized in the real time phase and update the corresponding commitment values $COM_r$, $COM_s$, $COM_{\pi}$, $COM_{D_{d_i}}$ to the commitment chain $\mathcal{C}_{COM}$ through the commitment transaction $tx_{COM}$.
	\item [2)] After receiving the blind messages $\boldsymbol{M}\times \boldsymbol{K}^{-1}$ from the users, the mix nodes invoke the algorithm $\Delta.\mathsf{RealTime}$ to mix the received messages and obtain the permuted message $\boldsymbol{\Pi}_n(\boldsymbol{M})$. 
\end{itemize}

\textbf{Audit}. In this phase, the users invoke $\Gamma^*.\mathsf{Verify}$ to check the integrity of the permuted messages. If the algorithm $\Gamma^*.\mathsf{Verify}$ returns 1 to all the users, then the permuted messages are considered integrated. If $\Gamma^*.\mathsf{Verify}$ returns 0 to some users, then the users ask the current mix nodes to check their calculation through the corresponding commitment block $B_{COM}$ and identify the malicious mix nodes. The malicious mix node will be removed out of BCMIX.
\section{Performance Evaluation}
This section presents the performance evaluation of BCMIX. We firstly theoretically analyze the number of candidate miners. Then we describe the implementation of BCMIX and evaluate its performance.
\begin{table*}[]
	\centering
	\scriptsize
	\caption{Time cost (IN RM s) for different number of mix nodes and users to execute the algorithms}
	\begin{threeparttable}
		\begin{tabular}{llllllllllllll}
			\Xhline{1pt}
			\multirow{2}{*}{\diagbox{Alg}{Time}{Number}}     &  \rule{0pt}{7pt} Mix Node   & \multicolumn{3}{c}{3} & \multicolumn{3}{c}{5} & \multicolumn{3}{c}{7}& \multicolumn{3}{c}{9} \\ \Xcline{2-14}{1pt} \rule{0pt}{8pt} 
			& \quad User       & 10    & 50    & 100    & 10    &50    & 100    & 10   & 50    & 100 &   10    &  50     & 100  \\ \Xhline{1pt}
			\multicolumn{2}{l}{\rule{0pt}{10pt}$\mathsf{PoWVote}$}      &  29.62    &   28.90   &  30.12      & 29.77      &   31.23   &  29.35    &  29.00   &   30.27   &  32.14   &  31.36     &   30.80    & 29.79 \\ \Xhline{1pt}
			\multicolumn{2}{l}{\rule{0pt}{10pt}$\mathsf{IPSharding}$ }    &    0.65   &   0.65    &     0.65   &   0.93    &    0.93   &   0.93     &   1.46    &   1.46     & 1.46  &  1.88    &  1.88      &   1.88    \\ \Xhline{1pt}
			\multicolumn{2}{l}{\rule{0pt}{10pt}Key Exchange}   &  0.60     &   0.83    &   1.07   &   0.69    &  1.03     &   1.47    &   0.77  &   1.26    &  1.86  &    0.86 &  1.48    & 2.26\\ \Xhline{1pt}
			\multicolumn{2}{l}{\rule{0pt}{10pt}$\mathsf{Precomputation}$} &   0.33    &   0.33   &   0.33   &    0.47   &  0.47    &   0.47    &   0.62    &    0.62   &  0.62 &     0.81  &   0.81    &  0.81   \\ \Xhline{1pt}
			\multicolumn{2}{l}{\rule{0pt}{10pt}$\mathsf{RealTime}$ }     &  0.03  &    0.11  &  0.26   &    0.06      &   0.19     &   0.37        &   0.12   &   0.23  & 0.43   &    0.16   &   0.29    &   0.49 \\ \Xhline{1pt}
		\end{tabular}
		\begin{tablenotes}
			\item[]We fix the difficulty $D=1\times10^{11}$ and stipulate the block generation time $T=30s$. We can change the division strategy in Figure \ref{fig6} to control the number of mix nodes. 
		\end{tablenotes}
	\end{threeparttable}
\end{table*}
\begin{figure}[]
	\centering
	\includegraphics[height=6cm, width=8.5cm]{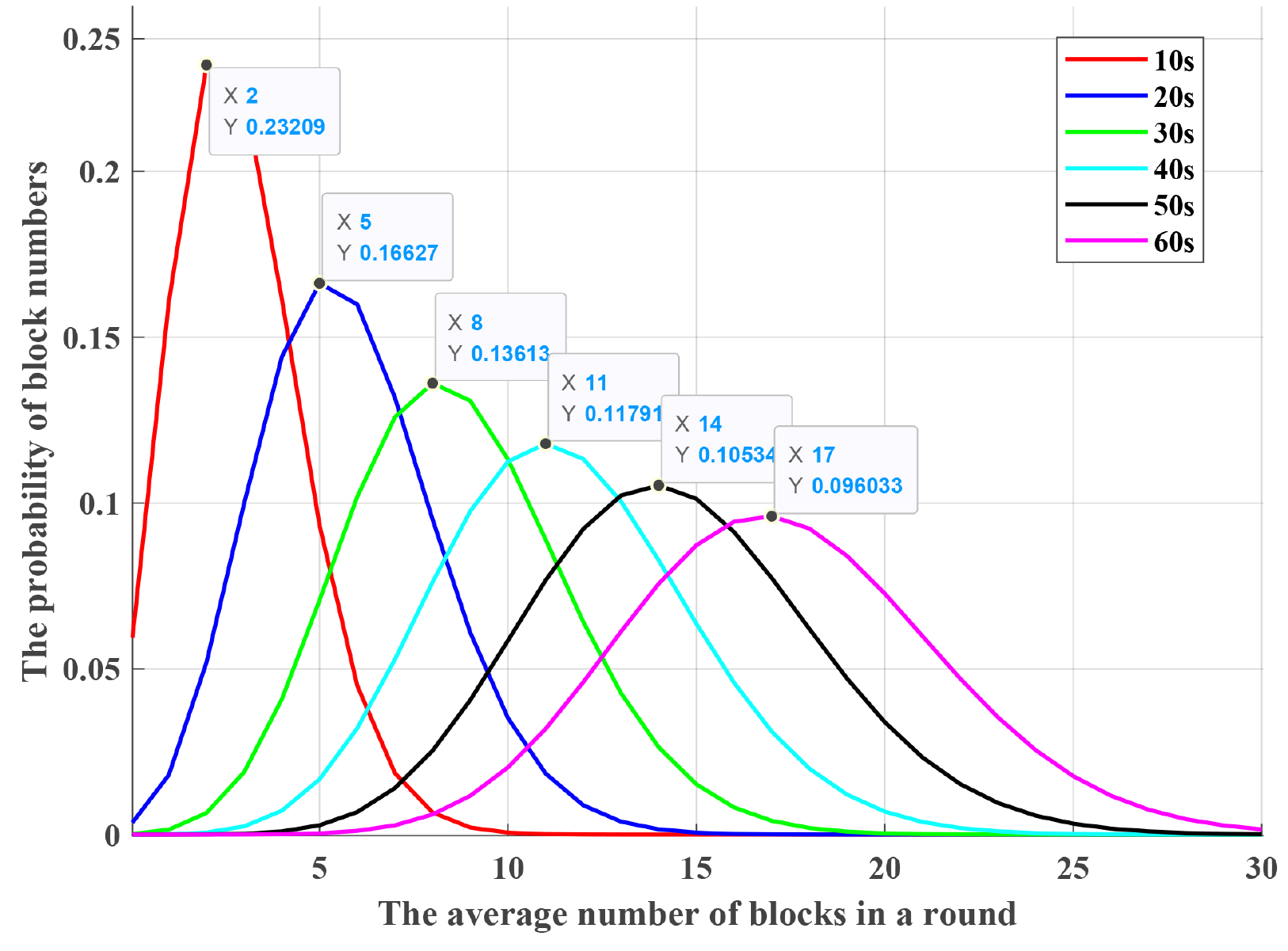}
	\caption{The probability of the expected number of candidate miners}
	\label{fig5}
\end{figure}
\subsection{The Number of candidate miners}
In Bitcoin, the difficulty value $D$, the target $target$ and the network hash rate satisfy the following formula:
$$D=\frac{D_{Max}}{target_{current}},  \quad Hashrate_{min}=\frac{D\cdot 2^{32}}{T},$$\\
where $D_{Max}$ is a large constant, $target_{current}$ denotes the current target and $Hashrate_{min}$ represents the minimum hash rate required to calculate a block of difficulty $D$ within time $T$. According to literature \cite{}, the process of generating $n$ blocks within the average time $t$ will be a Poission distribution with the expected value $\lambda$,
 $$P(X\geq M)=\sum_{k=M}^{\infty}\frac{\lambda^k}{k!}e^{-\lambda}.$$ 
 
Since the Bitcoin network with $Hashrate_{min}$ mining power generates one block within average time $T$, we can conclude that \\
$$\lambda=\frac{Hashrate_{real}}{Hashrate_{min}}=\frac{Hashrate_{real}\cdot T}{D\cdot 2^{32}}.$$
Here $Hashrate_{real}$ denotes the mining power of the Bitcoin network in the real world. To better simulate BCMIX in the real world scenario, we leverage the mining power distribution from \footnote{https://btc.com/stats/pool?pool\_mode=year}. In Appendix B, Table IV details the IP address and mining power of different mining pools. With $D=1\times10^{11}$, we estimate the probability of simultaneous block generation under different $T$ and report our results in Figure \ref{fig5}. We can see that when fixing the difficulty, the longer the average time to generate a block, the higher the probability of generating multiple blocks simultaneously.

\begin{figure}[]
	\centering
	\includegraphics[height=6cm, width=7cm]{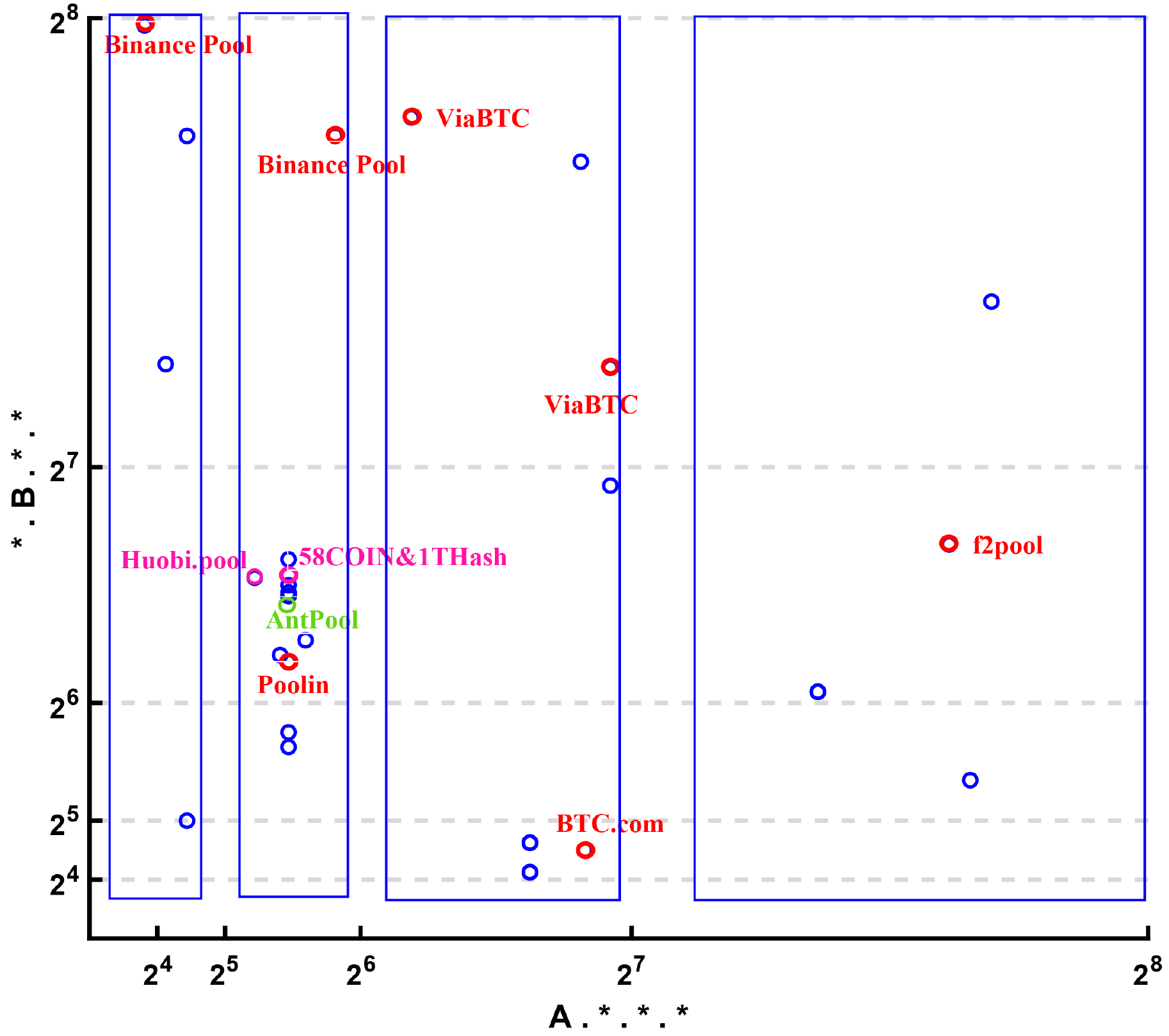}
	\caption{The results of $\mathsf{IPSharding}$ algorithm when $\sigma=5$ (four node poos, i.e. four mix nodes). We denote an IP address as four segments, the $x$ axis is the first segment and the $y$ axis denotes the second segment. We can change the division strategy to control the number of mix nodes. We respectively simulate the attackers which hold 53.84\% (the red circle), 63.43\% (the red circles and the green circle) and 74.80\% (the red circles, the green circle and the purple circles) of total computing powers.}
	\label{fig6}
\end{figure}
\subsection{Implementation}
In order to show the feasibility of BCMIX, we build a Bitcoin network in a desktop computer (equipped with a Ubuntu 16.04 LTS, Intel (R) Core (TM) i5-8500 CPU of 3.00GHz and 16GB RAM). We make use of two Github programs, Bitcoin-Simulator \footnote{https://github.com/arthurgervais/Bitcoin-Simulator} and cMix \footnote{https://github.com/byronknoll/cmix}, to evaluate the proposed BCMIX. We set up the public parameters by utilizing the publicly available library for cryptography on the $\mathsf{secp256k1}$ curve \footnote{https://github.com/bitcoin-core/secp256k1}.

We apply the Algorithm 2 to Table IV and obtain the distribution of mining pools in the real world. The results are shown in Figure 6. Combining Figure \ref{fig5} and Figure \ref{fig6}, we can find that when fixing difficulty $D$ to $1.0\times 10^{11}$, setting the average block generation time to 20s or 30s can obtain a reasonable number of mix nodes, which is more conducive to the implementation of BCMIX. To measure the approximate time cost, we test the functionality of $\mathsf{PoWVote}$, $\mathsf{IPSharding}$, Key Exchange, $\mathsf{Precomputation}$ and $\mathsf{RealTime}$ under different number of mix nodes and users for 100 times. The test results are shown in Table II. We can see that $\mathsf{PoWVote}$ spends more time than other algorithms. In order to improve the operating efficiency of BCMIX, we elect mix nodes at regular intervals.
\begin{table*}[]
	\centering
	\begin{threeparttable}
		\caption{Summary comparison of attack resilience and performance}
		\begin{tabular}{llccccccc}
			\Xhline{1pt}
			\multirow{2}{*}{} & \multirow{2}{*}{\textbf{Technology}} & \multicolumn{5}{c}{\textbf{Attacks}}                     & \multicolumn{2}{c}{\textbf{Performance}} \\ \Xcline{3-9}{1pt}\rule{0pt}{10pt}
			&                             & \textbf{Traffic analysis} & \textbf{DDos} & \textbf{Tagging} & \textbf{MitM} & \textbf{Sybil} &\textbf{Bandwidth}        & \textbf{Latency}       \\\Xhline{1pt}
			BCMIX             & Blockchain Re-encryption                 &      \Checkmark             &   \Checkmark    &   \Checkmark       &   \Checkmark    &  \Checkmark      &       \Checkmark            &     \Checkmark           \\\Xhline{1pt}
			cMix              & Mix Re-encryption           &       \Checkmark          &   \XSolidBrush   &    \XSolidBrush     &   \XSolidBrush   &  $\mathbf{\setminus}$  &            \Checkmark      &      \Checkmark         \\ \Xhline{1pt}
			Tor               & Mix Mutilayer Encryption   &     \XSolidBrush             &   \XSolidBrush   &  $\setminus$      &   \XSolidBrush   &   $\setminus$    &          \Checkmark         &   \Checkmark             \\ \Xhline{1pt}
			BAR               & Multicast/Broadcast         &         \Checkmark         &    \XSolidBrush   &    \XSolidBrush      &   \XSolidBrush    &   $\setminus$    &       \Checkmark            &       \XSolidBrush          \\ \Xhline{1pt}
			AnonPubSub        & Probabilistic Forwarding    &   \Checkmark                &    \XSolidBrush  &     \Checkmark    & \XSolidBrush     &     $\mathbf{\setminus}$   &        $\setminus$           &  $\setminus$              \\ \Xhline{1pt}
			Tarzan            & Peer-to-peer                &       \Checkmark            &   \Checkmark    & \Checkmark        &     \XSolidBrush  &    \XSolidBrush    &     \Checkmark              &    \Checkmark            \\ \Xhline{1pt}
			DiceMix           & CoinJoin                    &      \XSolidBrush             &    \XSolidBrush   &    \Checkmark     &    \XSolidBrush   &   \XSolidBrush     &    \Checkmark              &      \Checkmark         \\ \Xhline{1pt}
		\end{tabular}
		\begin{tablenotes}
			\item[]\Checkmark - The system is secure against this attack or the system performs well.
			\item[] \XSolidBrush - The system is vulnerable to this attack or the system is not performing well. 
			\item[] $\setminus$ - The system does not involve this attack of the authors did not mention the related performance.
		\end{tablenotes}
	\end{threeparttable}
\end{table*}
\section{Security Analysis}
According to the proposed additive homomorphism mix-net protocol and basic components of proof-of-work blockchains, BCMIX system can satisfy all the security requirements described in Section \uppercase\expandafter{\romannumeral3}-C. Table III summarizes a comparison among different anonymous systems.
\begin{figure}[]
	\centering
	\includegraphics[height=4.5cm, width=8.5cm]{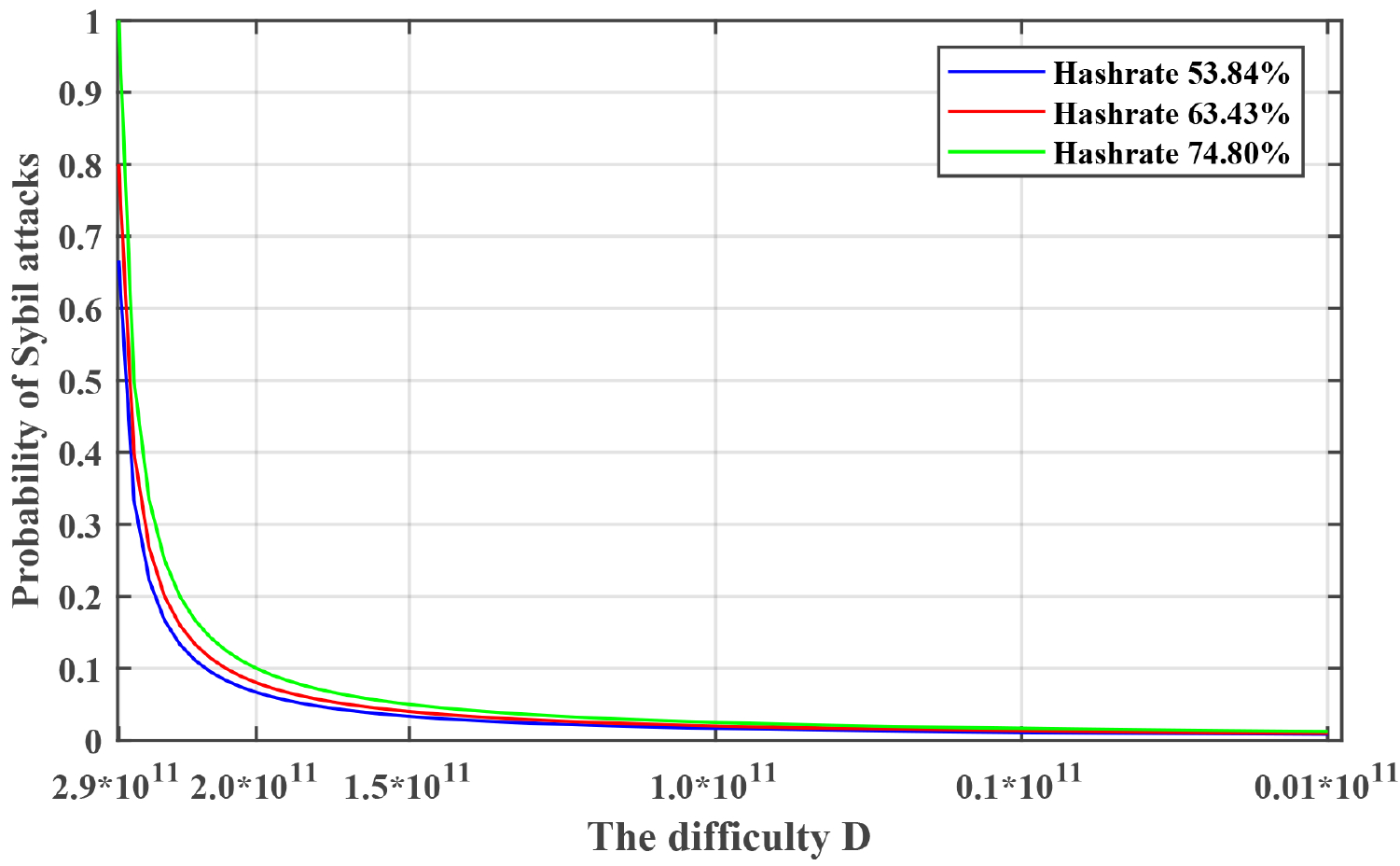}
	\caption{The relationship between difficulty and probability to launch the Sybil attacks.}
		\label{fig7}
\end{figure}
\begin{itemize}
	\item {\textbf{Resistance to Sybil Attacks.}} We indicate that an attacker implement Sybil attacks successfully means the attacker take control of all mix nodes. To simulate a process of Sybil attacks, we assume that an IP address of mining pools in Table IV represents an identity, and all identities created by a mining pool share the mining pool's computing power equally. Figure \ref{fig7} shows the requirements in terms of difficulty (computing power), in order for an attacker with different computing power to successfully launch Sybil attacks. Figure \ref{fig7} indicates that, when the difficulty $D$ is fixed, an attacker with higher computing power have a higher probability of successfully launching Sybil attacks than attackers with lower computing power. And if we set the difficulty $D$ to be very small, the probability that an attacker successfully launching Sybil attacks is almost zero.	
	   \begin{figure}[]
		\centering
		\includegraphics[height=4.5cm, width=8.5cm]{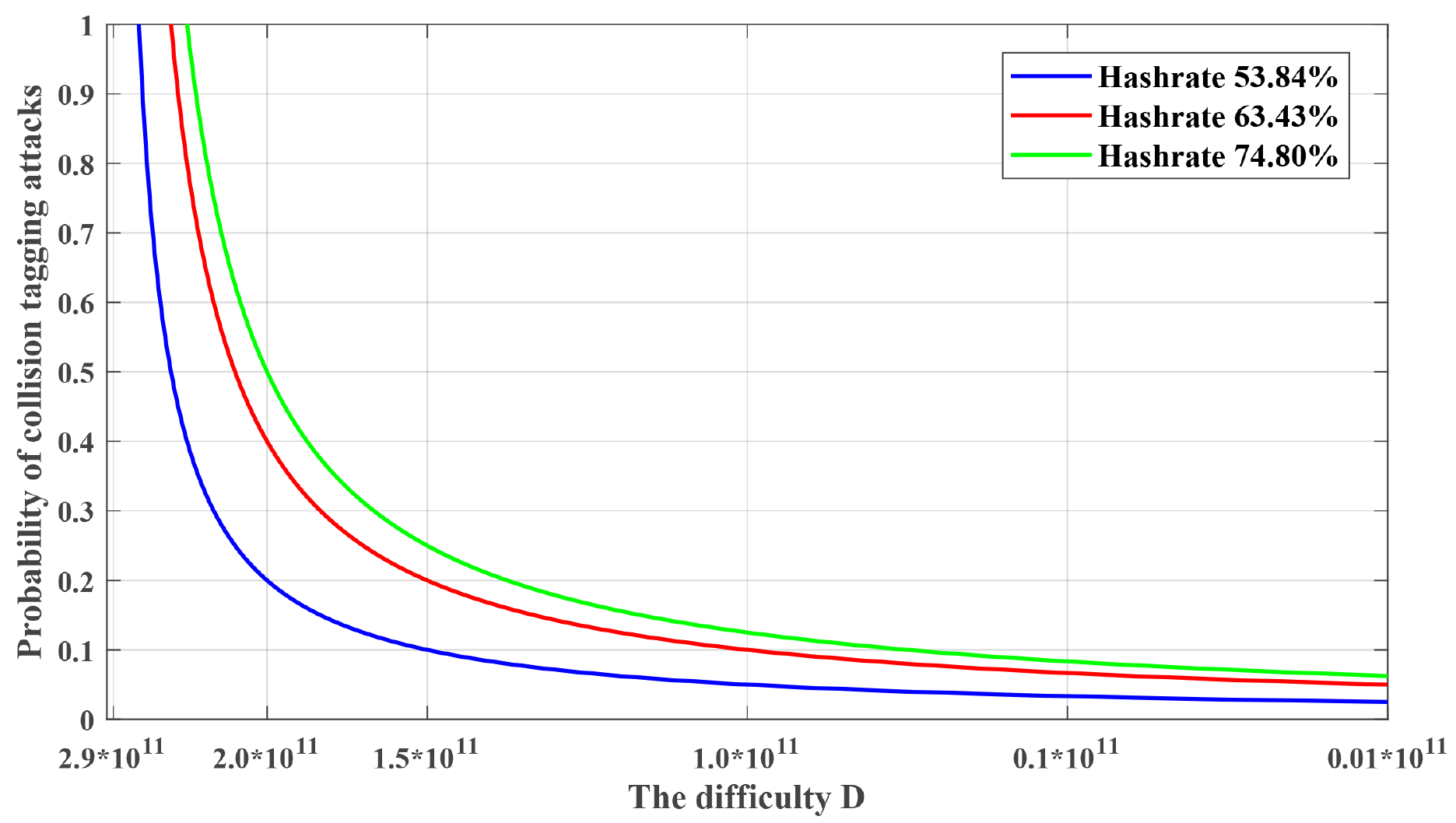}
		\caption{The relationship between difficulty and probability to launch the collision tagging attacks.}
		\label{fig8}
	\end{figure}
	\item {\textbf{Resistance to Collision Tagging Attacks.}} As we mention in Section II C, an attacker compromise the last mix node $N_l$ and any mix node $N_i$ to launch collision tagging attacks. In this case, we can treat a collision tagging attack as a special form of a Sybil attack. We illustrates the relationship between the difficulty $D$ and the probability of launching collision tagging attacks in Figure \ref{fig8}. Similar to Sybil attacks, we can set the difficulty $D$ to be very small to resist the collision tagging attacks.
	\item {\textbf{Sender Anonymity.}} Suppose that an adversary $\mathcal{A}$ can compromise $\beta-2$ users and $n-1$ nodes. Let $U_x$ and $U_y$ denote any two honest users and let $N_i$ be the only honest mix node. The sender anonymity property guarantees that the adversary cannot distinguish the messages from the two honest users. We design the following experiments for a $\mathcal{PPT}$ adversary $\mathcal{A}$.\\\\
	$\boldsymbol{EXP_{an}^b(\Delta,\mathcal{A},\lambda)}:$\\
	$\mathcal{A}^\textbf{Setup}$:
	\begin{itemize}
	\item []$({\boldsymbol{Q}},K_x,K_y,K_{m,m\neq x,m\neq y}^{\mathcal{A}})\leftarrow \mathsf{Setup}(1^\lambda),$
    \end{itemize}
	$\mathcal{A}^\textbf{Precomputation}$:
	\begin{itemize}
    \item []	$(\mathcal{E}_{\boldsymbol{Q}}(\boldsymbol{R}_\mathcal{A})\cdot \mathcal{E}_{\boldsymbol{Q}}(\boldsymbol{r}_{N_i}))\leftarrow \mathsf{Preprocess}(\boldsymbol{r}),$
    \item [] $\mathcal{E}_{\boldsymbol{Q}}(\boldsymbol{\Pi}_n(\boldsymbol{R}_n)\cdot\boldsymbol{S}_n)\leftarrow\mathsf{Mix}(\boldsymbol{s},\pi),$
	\item []$(\boldsymbol{\Pi}_n(\boldsymbol{R}_n)\cdot\boldsymbol{S}_n)\leftarrow\mathsf{Postprocess}(D_{d_i}),$
	\end{itemize}
	$\mathcal{A}^\textbf{RealTime}$:\\
	$(m_b\cdot \boldsymbol{r}_{N_i},m_{1-b}\cdot \boldsymbol{r}_{N_i})\leftarrow\mathsf{Preprocess}(K_x\cdot m_b,K_y\cdot m_{1-b}), $
	\begin{itemize}
	\item []$\pi_{N_i}(m_b\cdot \boldsymbol{r}_{N_i}\cdot \boldsymbol{s}_{N_i},m_{1-b}\cdot \boldsymbol{r}_{N_i}\cdot \boldsymbol{s}_{N_i})\leftarrow \mathsf{Mix},$
    \item []\quad\quad\quad$(m_b^{'},m_{1-b}^{'})\leftarrow \mathsf{Postprocess}$
    \end{itemize}

   The adversary's advantage in the experiments is:
   $$|Pr[EXP_{an}^0(\Delta,\mathcal{A},\lambda)=1]-Pr[EXP_{an}^1(\Delta,\mathcal{A},\lambda)=1]|.$$

   \begin{definition}(Anonymity). 
	An additive homomorphism mix-net protocol $\Delta:=(\mathsf{Setup},\mathsf{Precom},\mathsf{RealTime})$ maintains anonymity if the advantage of the adversary in the anonymity game is negligible.
   \end{definition}
  \begin{theorem}
   If $\mathcal{E}$ is a ECDLP-secure additive homomorphism encryption scheme, ECDH is a ECDHP-secure key exchange protocol and a non-interactive commitment scheme $COM$ is perfectly-hiding, then BCMIX satisfies anonymity defined in Definition 7.
   \end{theorem}
 \begin{proof}(Sketch). We prove the security of BCMIX by reduction from the security of the encryption system $\mathcal{E}$. Without loss of generality, we assume that the adversary $\mathcal{A}$ can compromise $\beta-2$ users and $n-1$ nodes. Let $U_x$ and $U_y$ denote any two honest users and let $N_i$ be the only honest mix node. 
		
	In the setup phase, $\mathcal{A}$ can control all the shared keys except $K_x$, $K_y$. In the precomputation phase, $\mathcal{A}$ gets command of random values $\boldsymbol{R}_{\mathcal{A}},\boldsymbol{S}_{\mathcal{A}}$ and random permutations $\boldsymbol{\Pi}_{\mathcal{A}}$ except $\boldsymbol{r}_{N_i},\boldsymbol{s}_{N_i},\pi_{N_i}$, the corresponding ciphertext $\mathcal{E}(\boldsymbol{r}_{N_i})$ and the decrypt share $D_{N_i}$. In the real time phase, users send blind messages to mix networks in the form of $\boldsymbol{M}\times\boldsymbol{K}$. The adversary $\mathcal{A}$ can parse the blind messages as $(K_x\cdot m_b,K_y\cdot m_{1-b})$. During the mixing process of real time phase, $\mathcal{A}$ decrypt the mixed messages and obtain the mixed messages $(m_b\cdot \boldsymbol{r}_{N_i}\cdot \boldsymbol{s}_{N_i},m_{1-b}\cdot \boldsymbol{r}_{N_i}\cdot \boldsymbol{s}_{N_i})$. Finally, $\mathcal{A}$ observe the plaintext messages of the form $(m_b^{'},m_{1-b}^{'})$.
	
	We assume a challenger $\mathcal{C}$ assigns the blind messages $(K_x\cdot m_0,K_y\cdot m_1)$, where $m_0=m_b$ and $m_1=m_{1-b}$ as determined by a random bit $b$. An adversary holding $(m_b^{'},m_{1-b}^{'})$ predicts the association between $(K_x\cdot m_0,K_y\cdot m_1)$ and $(m_b^{'},m_{1-b}^{'})$. In the end, if the anonymity game adversary $\mathcal{A}$ predicts the bit $b$ correctly, we can infer that $\mathcal{A}$ can calculate $K_x$ and $K_y$ which break the ECDHP-secure of the underlying key exchange system. Thus $|Pr[EXP_{an}^0(\Delta,\mathcal{A},\lambda)=1]-Pr[EXP_{an}^1(\Delta,\mathcal{A},\lambda)=1]|\leq negl_{ECDHP}$.  
	Similarly, we conclude that $|Pr[EXP_{an}^0(\Delta,\mathcal{A},\lambda)=1]-Pr[EXP_{an}^1(\Delta,\mathcal{A},\lambda)=1]|\leq negl_{ECDLP}$. This concludes the proof. 
	\end{proof}
	\item {\textbf{Resistance to MitM attacks.}} In BCMIX we leverage the gossip protocol to spread messages. We assume that an attack cannot control all access networks of a blockchain node. According to \cite{heilman2015eclipse}, this is reasonable since none instances of eclipse attacks have arisen in reality up to now. Besides, many blockchain communities have fixed this vulnerability \cite{marcus2018low}. In this case, an adversary $\widetilde{A}$ attempts to simulate an elected mix node to deceive a user $\widetilde{U}$. At the same time, other nodes connected with $\widetilde{U}$ inform $\widetilde{U}$ the latest elected mix node sets, thus $\widetilde{U}$ identities $\widetilde{A}$ as an adversary. 
	
 	\item {\textbf{Single point of failure Resilience.}} Nodes in BCMIX are networking dynamically. Once an elected mix node $N_i$ loses the response for a period of time, the next mix node $N_{i+1}$ would inform other nodes that $i$ is crashed. Then the other nodes validate the situation of node $i$, and begin a new networking process if $i$ goes down, or identify node $N_{i+1}$ as a malicious node if $i$ works normally.	
	\item {\textbf{Resistance to Other Attacks.}} BCMIX can also resist the following attacks. 
	\begin{itemize}
		\item [a)] Replay attacks. An attacker may retransmitting a message $m$ form a previous session. Then the attacker compare the mixed message sets with the previous message sets which contain $m$, thus the attacker can associate the egress messages with the ingress messages. Since the random values and permutations are never reused, thus BCMIX resists replay attacks.
		\item [b)] Traffic analysis attacks. In connection-based systems such as Tor, attackers can distinguish between two different paths in the free mix network by counting packages and timing communication. Since BCMIX is a message-based system which batches and permutes messages during the transmitting process, attackers can not distinguish and analyze the blind messages, thus BCMIX resists traffic analysis attacks.
		\item [c)] Intersection attacks and statistical disclosure attacks. These attacks utilize information given by observing mix networks where the users can freely choose the mix node for their messages. Since BCMIX adopts a fixed cascade of mix nodes every round, thus BCMIX is not susceptible to these attacks.
	\end{itemize}
\end{itemize}
\section{Conclusion}
In this paper, we achieved a dynamic self-organizing mix anonymous system. With blockchain technology, we elect mix nodes from public, dynamic blockchain miners. Before constructing BCMIX, we proposed BCMN protocol with the formal security models. Building on the proposed protocol, we designed a transaction-based key exchange scheme and proposed our BCMIX. Then we illustrated BCMIX can satisfy the relevant security requirements. After that we demonstrated experimentally that BCMIX is resistant to the attacks proposed in this paper.
Finally, we evaluated the performance of the prototype with the real world data and compared BCMIX with some latest anonymous systems. The results suggested that BCMIX is practical for real world deployment.

A follow-on work is to find a solution for recipient anonymity, which would improve the anonymity ability of our system. We believe that 
building a bidirectional anonymous system allow us to identity additional features and properties.

\section{Appendix A.}
The detailed description of cMix protocol is as follows.

\textbf{Setup phase.} The mix nodes establish their decryption share $X_i$, and the public key $y$ is computed. Each user $A_j$ will individually establish a symmetric key $k_{i,j}$ with each mix node $N_i$ in the network. The mix nodes draw their random values $\vec{r}_i$ and $vec{t}_i$ for the $n$ slots.

\textbf{Precomputation phase.} The goal in this phase is to perform the public-key operations that is needed in the real time phase. 

\textit{Step 1-Preprocessing}: Mixnode $N_i$ computes $\mathcal{E}(\vec{r}_i^{-1})$, and send their calculated vector to the network handler. The network handler then computes $\mathcal{E}(\vec{R}_h^{-1})=\prod_{i=1}^{h}\mathcal{E}(\vec{r}_i^{-1})$.

\textit{Step 2-Mixing}: $N_i(i=1,\cdots,h-1)$ computes and sends the following to $N_{i+1}$：
\begin{scriptsize}$$\mathcal{E}(\Pi_i(\vec{R}_h^{-1})\times \vec{T}_i^{-1})=\begin{cases}
\pi_1(\mathcal{E}(\vec{R}_h^{-1}))\times \mathcal{E}(\vec{t}_1^{-1})& i=1\\
\pi_i(\mathcal{E}(\Pi_{i-1}(\vec{R}_h^{-1})\times \vec{T}_{i-1}^{-1}))\times \mathcal{E}(\vec{t}_i^{-1}) & 1<i\leq h 
\end{cases}$$
\end{scriptsize}
$N_h$ finally computes: $(\vec{C}_1,\vec{C}_2)=\mathcal{E}(((\vec{R}_h)\times\vec{T}_h)^{-1})=\pi_h(\mathcal{E}((\vec{R}_h^{-1})\times \vec{T}_{h-1}^{-1}))\times \mathcal{E}(\vec{t}_h^{-1})$. $N_h$ sends $\vec{C}_1$ to the other mix nodes and store $\vec{C}_2$ locally for use in the real time phase.

\textit{Step 3-Postprocessing}: Mixnode $N_i$ use their decryption share $X_i$ to decrypt the vector of random components they received in the previous step; $\mathcal{D}_i(\vec{C}_1)=\vec{C}_1^{-X_i}$. They publish a commitment to their calculated decryption share.

\textbf{Real time phase.} In this phase, the senders are involved. $A_j$ constructs a blinded message $m_j\times K_j^{-1}$. The blnded messages are the input to the protocol, and they are combinded by the network handler to yield the vector $\vec{m}\times \vec{K}^{-1}$.

\textit{Step 1-Preprocessing}: Every mix node $N_i$ calculates $\vec{k}_i\times \vec{r}_i$, and sends the resulting vector to the network handler. The network handler then computes $\vec{m}\times\vec{R}_h=\vec{m}\times \vec{K}^{-1}\times \prod_{i=1}^{h}\vec{k}_i\times \vec{r}^i$, hence the $\vec{K}^{-1}$ vector is replaced with the random $r$ values of each mix node.

\textit{Step 2-Mixing}: $N_i$ computes and sends the following to $N_{i+1}$:
\begin{scriptsize}
$$(\Pi_i(\vec{m}\times\vec{R}_h)\times \vec{T}_i)=\begin{cases}
\pi_1(\vec{m}\times\vec{R}_h)\times\vec{t}_1 & i=1\\
\pi_i(\Pi_{i-1}(\vec{m}\times\vec{R}_h)\times\vec{T}_{i-1})\times\vec{t}_i & 1<i<h
\end{cases}$$
\end{scriptsize}
Mix node $N_h$ computes $\Pi_h(\vec{m}\times\vec{R}_h)\times\vec{T}_h=\pi_h(\Pi_{h-1}(\vec{m}\times\vec{R}_h)\times\vec{T}_{H-1})\times\vec{t}_h$. $N_h$ commits to this vector and sends the commitment to the remaining mix nodes.

\textit{Step 3-Postprocessing}: When mix node $N_i(i=1,\cdots,h-1)$ receive the commitment from $N_h$, they send their decryption share $\mathcal{D}_i(\vec{C}_1)$ computed in the precomputation phase to the network handler. The last mix node computes and send the following to the network handler: $\Pi_h(\vec{m}\times\vec{R}_h)\times\vec{T}_h\times\vec{C}_2\times\prod_{i=1}^{h}\mathcal{D}(\vec{C}_1)=\Pi_h(\vec{m}\times\vec{R}_h)\times\vec{T}_h\times(\Pi_h(\vec{R}_h)\times\vec{T}_h)^{-1}=\Pi_h(\vec{m})$

The network handler outputs $\Pi_h(\vec{m})$, that is a permutation of the input message.
\section{Appendix B.}
The distribution of mining pools in the real world. To simulate a process of Sybil attacks, we assume that an IP address of mining pools in Table IV represents an identity, and all identities created by a mining pool share the mining pool's computing power equally. 
\begin{table}[htbp]
	\centering\caption{Notations of The Revised Mix-net Protocol}
	\scriptsize
	\label{tab2}
	\begin{tabular}{cccc}
		\Xhline{0.7pt}
		\textbf{Pool Name}      & \textbf{IP Address}                                                                                              & \textbf{Hashrate} & \textbf{Proportion} \\ 	\Xhline{0.7pt}
		&                                                                                                         & 124(EH/s)            & 100\%       \\ 	\Xhline{0.7pt}
		f2pool         & 203.107.32.162                                                                                          & 21.1048        & 17.02\%     \\ 	\Xhline{0.7pt}
		Poolin         & 47.75.234.12                                                                                            & 19.716         & 15.9\%      \\ 	\Xhline{0.7pt}
		BTC.com        & \begin{tabular}[c]{@{}c@{}}117.24.1.243, 117.24.1.239\\ 117.24.1.238, 117.24.1.242\end{tabular} & 16.1076        & 12.99\%     \\ 	\Xhline{0.7pt}
		AntPool        & 47.94.135.145                                                                                           & 13.95          & 11.25\%     \\ 	\Xhline{0.7pt}
		ViaBTC         & 116.211.155.211, 123.155.158.10                                                                         & 7.8988         & 6.37\%      \\ 	\Xhline{0.7pt}
		Huobi.pool     & 47.93.94.105                                                                                            & 7.626          & 6.15\%      \\ 	\Xhline{0.7pt}
		58COIN\&1THash & 39.98.72.224                                                                                            & 6.4728         & 5.22\%      \\ 	\Xhline{0.7pt}
		SlushPool      & \begin{tabular}[c]{@{}c@{}}104.26.5.102, 104.26.4.102\\ 172.67.74.105\end{tabular}                  & 5.6792         & 4.58\%      \\ 	\Xhline{0.7pt}
		OKExPool       & 208.43.170.231                                                                                          & 4.9724         & 4.01\%      \\ 	\Xhline{0.7pt}
		unknown        & 47.93.94.105                                                                                            & 4.7244         & 3.81\%      \\ 	\Xhline{0.7pt}
		BTC.TOP        & 123.56.208.222                                                                                          & 3.9804         & 3.21\%      \\ 	\Xhline{0.7pt}
		BytePool       & 58.218.215.133                                                                                          & 2.2816         & 1.84\%      \\ 	\Xhline{0.7pt}
		Binance Pool   & 13.248.150.68, 76.223.2.151                                                                          & 2.0832         & 1.68\%      \\ 	\Xhline{0.7pt}
		BitFury        & \begin{tabular}[c]{@{}c@{}}104.26.5.32, 104.26.4.32\\ 172.67.70.128\end{tabular}                   & 2.0336         & 1.64\%      \\ 	\Xhline{0.7pt}
		Lubian.com     & 47.56.109.242                                                                                           & 1.922          & 1.55\%      \\ 	\Xhline{0.7pt}
		NovaBlock      & \begin{tabular}[c]{@{}c@{}}104.26.11.113, 104.26.10.113\\ 172.67.70.128\end{tabular}                 & 1.488          & 1.20\%      \\ 	\Xhline{0.7pt}
		SpiderPool     & 47.52.126.9                                                                                             & 0.62           & 0.5\%       \\ 	\Xhline{0.7pt}
		WAYI.CN        & 47.103.164.189                                                                                          & 0.5459         & 0.44\%      \\ 	\Xhline{0.7pt}
		Bitcoin.com    & 104.18.26.217, 104.18.27.217                                                                         & 0.4092         & 0.33\%      \\ 	\Xhline{0.7pt}
		MiningCity     & 23.218.94.192, 23.32.241.177                                                                         & 0.124          & 0.1\%       \\ 	\Xhline{0.7pt}
		OKKONG         & 47.96.193.193                                                                                           & 0.062          & 0.05\%      \\ 	\Xhline{0.7pt}
		TATMAS Pool    & \begin{tabular}[c]{@{}c@{}}104.18.40.151, 172.67.148.229\\ 104.1841.151\end{tabular}                 & 0.0496         & 0.04\%      \\ 	\Xhline{0.7pt}
		BitClub        & 213.173.105.14                                                                                          & 0.0496         & 0.04\%      \\ 	\Xhline{0.7pt}
		Sigmapool.com  & 18.156.81.156                                                                                           & 0.0372         & 0.03\%      \\  	\Xhline{0.7pt}
		KanoPool       & 45.77.7.149                                                                                             & 0.0124         & 0.01\%      \\ 	\Xhline{0.7pt}
		Solo CK        & 51.81.56.15                                                                                             & 0.0124         & 0.01\%      \\ 	\Xhline{0.7pt}
	\end{tabular}
\end{table}


\bibliographystyle{IEEEtran}
\bibliography{reference}


\begin{IEEEbiography}{Michael Shell}
Biography text here.
\end{IEEEbiography}

\begin{IEEEbiographynophoto}{John Doe}
Biography text here.
\end{IEEEbiographynophoto}


\begin{IEEEbiographynophoto}{Jane Doe}
Biography text here.
\end{IEEEbiographynophoto}

\end{document}